\input ./style/arxiv-vmsta.cfg
\documentclass[numbers,compress,v1.0.1]{vmsta}

\usepackage{enumerate}
\usepackage{dcolumn}

\volume{3}
\issue{1}
\pubyear{2016}
\firstpage{19}
\lastpage{45}
\doi{10.15559/16-VMSTA47}% Updated by VTEXPTS2LaTeX.exe, 25.02.2016 10:41

\DeclareMathOperator{\ME}{{\sf E}}
\DeclareMathOperator{\Prob}{{\sf P}}
\DeclareMathOperator{\rank}{{rank}}

\DeclareMathOperator{\supp}{{supp}}
\DeclareMathOperator{\betatolines}{{lob}}
\DeclareMathOperator{\argmin}{{arg\/min}}
\DeclareMathOperator{\argmax}{{arg\/max}}
\DeclareMathOperator{\ellell}{{\ell\ell}}
\DeclareMathOperator{\hatellell}{{\widehat{\ell\ell}}}
\DeclareMathOperator{\sigmatwo}{{sigma2}}
\DeclareMathOperator{\vectorize}{{vec}}
\newcommand{\mathbi}[1]{\mbox{\boldmath$#1$}}
\newcommand{\IgnoreF}{Ignore-$\widehat{F}$}
\def\eatdot.{}

\newcommand{\boldbeta}{\mbox{\boldmath$\beta$}}
\newcommand{\hatboldbeta}{\hat{\boldsymbol\beta}}
\newcommand{\boldmu}{\mbox{\boldmath$\mu$}}
\newcommand{\Psiinf}{\overline{\boldsymbol\varPsi}_\infty}
\newcommand{\dPsiinf}{\overline{\boldsymbol\varPsi}{}'_\infty}
\newcommand{\ddPsi}{{\boldsymbol\varPsi}{}''_1}
\newcommand{\betanormaltrue}{\boldbeta_{{\rm tn}}}
\newcommand{\distrto}{\stackrel{{\sf d}}{\longrightarrow}}
\newcommand{\xiold}{{\xi\,{\rm old}}}

\theoremstyle{plain}% Theorems
\newtheorem{thm}{Theorem}
\newtheorem{prop}[thm]{Proposition}
\theoremstyle{remark}% Remark
\newtheorem{remark}{Remark}

\startlocaldefs

\newcommand{\rrVert}{\Vert}
\newcommand{\llVert}{\Vert}
\newcommand{\rrvert}{\vert}

\allowdisplaybreaks
\endlocaldefs

\begin{document}
\begin{frontmatter}

\title{Equivariant adjusted least squares estimator in two-line fitting model}

\author{\inits{S.}\fnm{Sergiy}\snm{Shklyar}}\email{shklyar@univ.kiev.ua}
\address{Taras Shevchenko National University of Kyiv, Ukraine}

\markboth{S. Shklyar}{ALS estimator in \xch{two-line}{two lines} fitting model}

\begin{abstract}
We consider the two-line fitting problem.
True points lie on two straight lines and are observed with Gaussian perturbations.
For each observed point, it is not known on which line the corresponding true point lies.
The parameters of the lines are estimated.

This model is a restriction of the conic section fitting model
because a couple of two lines is a degenerate conic section.
The following estimators are constructed:
two projections of the adjusted least squares estimator in the conic section fitting model,
orthogonal regression estimator,
parametric maximum likelihood  estimator in the Gaussian model,
and regular best asymptotically normal moment estimator.

The conditions for the consistency and asymptotic normality of the projections
of the adjusted least squares estimator are provided.
All the estimators constructed in the paper are equivariant.
The estimators are compared numerically.
\end{abstract}


\received{30 January 2016}% Updated by VTEXPTS2LaTeX.exe, 25.02.2016 10:41
\revised{19 February 2016}% Updated by VTEXPTS2LaTeX.exe, 25.02.2016 10:41
\accepted{19 February 2016}% Updated by VTEXPTS2LaTeX.exe, 25.02.2016 10:41
\publishedonline{21 March 2016}
\end{frontmatter}

\section{Introduction}\label{sec-intro}
\subsection{Two-line fitting model}\label{ss:model2lines}
Consider a problem of estimation of two lines
by perturbed observations of points that lie on the lines.
Let the true points $(\xi_i, \eta_i)$
lie on the union of two different lines $\eta = k_1 \xi + h_1$
and $\eta = k_2 \xi + h_2$, that is,
\begin{equation}
\biggl[\begin{array}{@{}l@{\quad}l}
\mbox{either} & \eta_i = k_1 \xi_i + h_1, \\
\mbox{or}  & \eta_i = k_2 \xi_i + h_2,
\end{array} %]
\quad i=1,2,\ldots\, . \label{q11a}
\end{equation}
Let these points be observed with perturbations $(\delta_i, \varepsilon_i)$,
$i=1,\ldots,n$, that~is, the observed points are $(x_i, y_i)$, $i=1,\ldots,n$,
with
\begin{gather}
x_i = \xi_i + \delta_i, \label{q11b1}
\\
y_i = \eta_i + \varepsilon_i.
\label{q11b2}
\end{gather}
The perturbations are assumed to be independent and identically
normally distributed,
\begin{equation}
\begin{pmatrix} \delta_i \\ \varepsilon_i \end{pmatrix} \sim N\bigl(0, \sigma^2 {\bf I}\bigr), \label{q11c}
\end{equation}
where $\bf I$ is the $2\times 2$ identity matrix.

The parameters $k_1$, $h_1$, $k_2$, $h_2$, and $\sigma^2$
are to be estimated.

We consider both functional and structural models.
In \emph{functional\/} model, the true points are assumed to be nonrandom.
In \emph{structural\/} model, the true points are assumed
to be independent and identically distributed (i.i.d.).
The errors $(\delta_i, \varepsilon_i)$ are i.i.d.\@
and independent of the true points.

In the structural model, $(\xi_i, \eta_i, \delta_i, \varepsilon_i)$
are i.i.d.\@ random vectors, and thus, the observed points $(x_i, y_i)$
are i.i.d.
In the functional model, the observed points are independent, Gaussian,
with different means but with common covariance matrix.

\begin{remark}
The true lines defined by Eqs.~\eqref{q11a} cannot
be parallel to the \hbox{$y$-axis}.
In order to avoid overflows during evaluation of the
estimators (except of RBAN-moment estimator),
another parameterization is used internally:
$\vec{\tau}^\top (\mathbi{\zeta} - \mathbi{\zeta}_0) = 0$,
where $\vec{\tau}$ is a unit vector orthogonal to the
line, and $\mathbi{\zeta}_0$ is a point on the line.
The computation of the RBAN-moment estimator
(see Section~\ref{ss:RBANmome}) is implemented for
explicit parameterization only.
Computational optimization of the RBAN-moment estimator
is a matter of further work.

The explicit parameterization has the advantage that the
number of parameters is equal to the dimension of parameter
space.
(In \cite{Shklyar2015p1}, the second-order equation
\eqref{q124ai} has six unknown coefficients, but the conic
section can be parameterized with five parameters.
The parameter space for the parameters of the conic section
was the five-dimensional unit sphere
in the six-dimensional Euclidean space.
Mismatch between the number of parameters
and the dimension of the parameter space
made the asymptotic covariance matrix of the estimator
\emph{singular}.)

In simulations, the confidence intervals for the coordinates
of the intersection point of the two lines are obtained
based on the asymptotic covariance matrix for the
intersection point.
For the projections the ALS2 estimator, that asymptotic
covariance matrix can be evaluated without use of explicit
line parameterization.
\end{remark}

\subsection{Conic section fitting model}
Let the true points $(\xi_i, \eta_i)$
lie on the second-order algebraic curve
\begin{equation}
A \xi_i^2 + 2 B \xi_i \eta_i
+ C \eta_i^2 + 2 D \xi_i + 2 E
\eta_i + F = 0, \quad i=1,2,\ldots. \label{q124ai}
\end{equation}
Hereafter, a second-order algebraic curve is called
a ``conic section'' or a ``conic.''

The points are observed with Gaussian perturbations,
and the perturbed points are denoted as $(x_i, y_i)$.
We have the same equations
\[
\begin{array}{l}
x_i = \xi_i  + \delta_i,      \\
y_i = \eta_i + \varepsilon_i,
\end{array} \quad \begin{pmatrix} \delta_i \\ \varepsilon_i \end{pmatrix} \sim N\bigl(0, \sigma^2 {\bf I}\bigr),
\]
\xch{as}{as Eqs.} \eqref{q11b1}--\eqref{q11c} in the two-line fitting model.

The vector of coefficients \xch{in}{in Eq.} \eqref{q124ai} is denoted by
$\boldbeta = (A, 2 B, C,\allowbreak 2D,\allowbreak 2E,\allowbreak F)^\top$.
The nonzero vector $\boldbeta$ and the error variance $\sigma^2$ are
the parameters of interest.

Similarly to the two-line fitting model,
the \emph{functional\/} and the \emph{structural\/}
models are distinguished.

A couple of lines is a degenerate case of a conic section.
Therefore, the conic section fitting model is an extension
of the two-line fitting model.

\subsection{ALS2 estimator in conic section fitting model}

We consider the adjusted least squares (ALS) estimator for
unknown $\sigma^2$.
The estimator is constructed in \cite{KukushMVH}.
Introduce the $6\times 6$ symmetric matrix
\begin{align*}
&\psi(x,y; v)
\\
& \newcommand{\dexxxx} {x^4{-}6x^2v{+}3v^2}
\newcommand{\dexxx} {x^3 - 3xv} \newcommand{\dexx} {x^2-v}
\newcommand{\deyyyy} {y^4{-}6y^2v{+}3v^2}
\newcommand{\deyyy} {y^3 - 3yv} \newcommand{\deyy} {y^2-v}
\newcommand{\dexxyy} {\bigl(x^2{-}v\bigr) \bigl(y^2{-}v
\bigr)} \quad {=}\begin{pmatrix}
\dexxxx        & (\dexxx)y   & (x^2{-}v)(y^2{-}v)  & *     & *     & * \\
(\dexxx)y      & (x^2{-}v)(y^2{-}v)& x\,(\deyyy) & *     & *     & * \\
(x^2{-}v)(y^2{-}v) & x\,(\deyyy) & \deyyyy     & *     & *     & * \\
\dexxx         & (\dexx)y    & x\,(\deyy)  & \dexx & xy    & x \\
(\dexx)y       & x\,(\deyy)  & \deyyy      & xy    & \deyy & y \\
\dexx          & xy          & \deyy       & x     & y     & 1
\end{pmatrix} .
\end{align*}
Asterisks are typed instead of some entries above the diagonal
of a symmetric matrix.
The entries of the matrix $\psi(x,y;v)$
are generalized Hermite polynomials in $x$ and $y$.
The matrix $\psi(x,y;v)$ is constructed such that
$\ME\psi(x_i,y_i;\sigma^2) = \psi(\xi_i,\eta_i;0)$
in the functional model and $\psi(\xi_i,\eta_i;0) \boldbeta = 0$
for the true points and true parameters.

Denote
\[
{\boldsymbol \varPsi}_n(v) = \sum_{i=1}^n
\psi(x_i, y_i; v) .
\]

The estimator $\hat\sigma^2$ of the error variance $\sigma^2$
is obtained from the equation
\begin{equation}
\label{q124g2} \lambda_{\min}\bigl({\boldsymbol \varPsi}_n\bigl(\hat
\sigma^2\bigr)\bigr) = 0 .
\end{equation}

Equation~\eqref{q124g2} always has a unique nonnegative solution.
If $n\ge 6$, then the solution to \eqref{q124g2} is positive almost surely.

The matrix ${\boldsymbol \varPsi}_n(\hat\sigma^2)$ is singular.
Define the estimator $\hatboldbeta$ of the vector $\boldbeta$
as a nonzero solution to the equation
\[
{\boldsymbol \varPsi}_n\bigl(\hat\sigma^2\bigr) \hatboldbeta = 0
.
\]

The strong consistency of the ALS2 estimator is proved in \cite{KukushMVH}
and  \cite{Shklyar2007} under somewhat different conditions.
The asymptotic normality is proved in \cite{Shklyar2015p1} for
the functional model and in \cite{Shklyar2015p2} for
the structural model.
Two consistent estimators of the asymptotic covariance matrix
are constructed in \cite{Shklyar2015p2}.

Denote
\begin{align*}
&\psi'_v(x,y;v) = \frac{\partial}{\partial v} \psi(x, y; v), \hspace{53pt} \boldsymbol\varPsi_1'' = \frac{\partial^2}{\partial v^2} \psi(x, y; v),\hspace{14pt}\\
&\hspace{173pt} \overline{\boldsymbol\varPsi}{}'_n(v)  = \frac{{\rm d}}{{\rm d}v} {\boldsymbol\varPsi}_n(v), \\
&\overline{\boldsymbol\varPsi}{}_n = \sum_{i=1}^n \psi(\xi_i,\eta_i; 0), \qquad \hspace{59pt} \overline{\boldsymbol\varPsi}{}'_n = \sum_{i=1}^n\psi'_v(\xi_i,\eta_i; 0),\\
&\overline{\boldsymbol\varPsi}{}_{\infty} = \lim_{n\to\infty}\frac{1}n \overline{\boldsymbol\varPsi}{}_n = \lim_{n\to\infty}\frac{1}n {\boldsymbol\varPsi}{}_n\bigl(\sigma^2\bigr), \qquad \overline{\boldsymbol\varPsi}{}_{\infty}' = \lim_{n\to\infty}\frac{1}n \overline{\boldsymbol\varPsi}{}'_n = \lim_{n\to\infty}\frac{1}n {\boldsymbol\varPsi}{}'_n \bigl(\sigma^2\bigr).
\end{align*}
Under the conditions of Proposition~\ref{prop:corec1} stated further, the latter limits exist
almost surely. See \cite{Shklyar2015p1} for explicit expressions of
the matrices $\psi'_v(x,y;v)$, ${\boldsymbol\varPsi}{}_1''$, $\overline{\boldsymbol\varPsi}{}_{\infty}$, and $\overline{\boldsymbol\varPsi}{}_{\infty}'$.
Note that $\ddPsi$ is a constant matrix.

\begin{prop}
\label{prop:corec1}
In the functional model, for all integer $p\ge 0$ and $q\ge 0$
such that $p+q\le 4$, let the following limits exist and be finite:
\[
\lim_{n\to\infty} \frac{1}{n} \sum
_{i=1}^n \xi_i^p
\eta_j^q =: \mu_{p,q},
\]
whereas in the structural model,
let $\ME \xi_1^4 < \infty$ and $\ME \eta_1^4 < \infty$.
In both models, let $\rank \Psiinf = 5$.
Then\textup{:}
\begin{enumerate}[\rm1.]
\item The estimator $\hatboldbeta$ is strongly consistent
in the following sense:
\begin{gather}
\label{eq:consHatbeta1} \min \biggl( \biggl\llVert \frac{\hatboldbeta}{\|\hatboldbeta\|} -
\frac{\boldbeta}{\|\boldbeta\|}\biggr\rrVert , \: \biggl\llVert \frac{\hatboldbeta}{\|\hatboldbeta\|} +
\frac{\boldbeta}{\|\boldbeta\|}\biggr\rrVert \biggr) \to 0 \quad \mbox{a.s.,}
\\
\label{eq:consHatsigma2} \hat\sigma^2 \to \sigma^2 \quad
\mbox{a.s.}
\end{gather}
\item  $\boldbeta{}^\top \dPsiinf \boldbeta < 0$.
\item  Eventually, $\hatboldbeta{}^\top {\boldsymbol\varPsi}{}'_n(\hat\sigma^2) \hatboldbeta < 0$.
\end{enumerate}
\end{prop}
``Eventually'' in the previous statement means that almost surely there exists $n_0$
such that  $\hatboldbeta{}^\top {\boldsymbol\varPsi}'_n(\hat\sigma^2) \hatboldbeta < 0$
for all $n \ge n_0$. In other words, almost surely,
$\hatboldbeta{}^\top {\boldsymbol\varPsi}'_n(\hat\sigma^2) \hatboldbeta \ge 0$
holds  only for  finitely many $n$.

Denote the normalized version of the true parameter
\[
\betanormaltrue = \sqrt{\displaystyle\frac{-1}{\boldbeta^\top \dPsiinf \boldbeta}} \boldbeta .
\]

Normalize the estimator of $\boldbeta$ in such a way that
$\widetilde{\mathbi{\beta}}{}^\top {\boldsymbol\varPsi}'_n(\hat\sigma^2)
 \widetilde{\mathbi{\beta}} = -n$
and $\boldbeta^\top \widetilde{\mathbi{\beta}} \ge 0$.
Therefore, denote
\begin{equation}
\label{eq:defbetatilde} \widetilde{\mathbi{\beta}} =
\begin{cases}
  \sqrt{\tfrac{ -n}
{\hatboldbeta{}^\top {\boldsymbol\varPsi}'_n(\hat\sigma^2) \hatboldbeta}} \,\hatboldbeta
& \mbox{if $\boldbeta{}^\top \hatboldbeta \ge 0$}, \\[6pt]
  -\sqrt{\tfrac{ -n}
{\hatboldbeta{}^\top {\boldsymbol\varPsi}'_n(\hat\sigma^2) \hatboldbeta}} \,\hatboldbeta
& \mbox{if $\boldbeta{}^\top \hatboldbeta < 0$}.
\end{cases}
\end{equation}

\begin{prop}
\label{prop:propALS2}
\textup{1.} Under the conditions of Proposition~\ref{prop:corec1},
the estimator $\widetilde{\mathbi{\beta}}$ is a strongly consistent estimator of
$\betanormaltrue = (-{\boldbeta^\top \dPsiinf \boldbeta})^{-1/2} \boldbeta$,
that is, $\widetilde{\mathbi{\beta}} \to \betanormaltrue$~a.s.

\textup{2.} In the functional model, for all integer $p\ge 0$ and $q\ge 0$
such that $p+q\le 6$, let the following limits exist and be finite:
\[
\lim_{n\to\infty} \frac{1}{n} \sum
_{i=1}^n \xi_i^p
\eta_j^q =: \mu_{p,q},
\]
whereas in the structural model,
let $\ME \xi_1^6 < \infty$ and $\ME \eta_1^6 < \infty$.
In both models, let $\rank \Psiinf = 5$.
Then the estimator  $\hat{\mathbi{\theta}}=(\widetilde{\mathbi{\beta}}{}^\top\!, \hat\sigma^2)^\top$
is asymptotically normal in the following sense:
\begin{equation}
\sqrt{n} \begin{pmatrix}
\widetilde{\mathbi{\beta}} - \betanormaltrue \\
\hat\sigma^2 - \sigma^2
\end{pmatrix} \distrto N(0, \, \varSigma_{\hat\theta}),
\end{equation}
where
\begin{align*}
&\varSigma_{\hat\theta} = \begin{pmatrix}
\Psiinf & \dPsiinf \betanormaltrue \\
\betanormaltrue^\top \dPsiinf &
  \frac{1}2 \betanormaltrue^\top \ddPsi \betanormaltrue
\end{pmatrix}^{\!\!{-}1} {\bf B} \begin{pmatrix}
\Psiinf & \dPsiinf \betanormaltrue \\
\betanormaltrue^\top \dPsiinf &
  \frac{1}2 \betanormaltrue^\top \ddPsi \betanormaltrue
\end{pmatrix}^{\!\!{-}1},
\\
&\quad {\bf B} = \lim_{n\to\infty} \frac{1}n \sum
_{i=1}^n \ME \mathbi{s}_i^{}
\mathbi{s}_i^\top,
\\
&\quad \mathbi{s}_i = \begin{pmatrix}
\psi(x_i, y_i; \sigma^2) \betanormaltrue \\
\frac{1}2 \betanormaltrue^\top
\psi(x_i, y_i; \sigma^2) \betanormaltrue - \frac{1}2
\end{pmatrix} .
\end{align*}

\textup{3.} Under the conditions of part \textup{2} of Proposition~\ref{prop:propALS2},
the following estimator of the asymptotic covariance matrix is consistent:
\begin{align*}
&\widehat{\varSigma}_{\hat\theta}(n) = {\bf A}^{-1}(n) {\bf B}(n) {
\bf A}^{-1}(n),
\\
&\quad {\bf A}(n) = \begin{pmatrix}
\frac{1}n {\boldsymbol\varPsi}_n(\hat\sigma^2) &
\frac{1}n {\boldsymbol\varPsi}'_n(\hat\sigma^2) \widetilde{\mathbi{\beta}} \\
\frac{1}n \widetilde{\boldbeta}^\top {\boldsymbol\varPsi}'_n(\hat\sigma^2) &
\frac{1}2 \widetilde{\boldbeta}^\top \ddPsi \widetilde{\boldbeta}
\end{pmatrix}, \quad {\bf B}(n) = \frac{1}n \sum
_{i=1}^n \hat{\mathbi{s}}_i^{}
\hat{\mathbi{s}}_i^\top,
\\
&\quad \hat{\mathbi{s}}_i = \begin{pmatrix}
\psi(x_i, y_i; \hat\sigma^2) \tilde{\mathbi{\beta}} \\
\frac{1}2 \tilde{\mathbi{\beta}}{}^\top
\psi(x_i, y_i; \hat\sigma^2) \tilde{\mathbi{\beta}} - \frac{1}2
\end{pmatrix} ,
\end{align*}
that is,
\[
\widehat{\varSigma}_{\hat\theta}(n) \to {\varSigma}_{\hat\theta}
\]
in probability.
\end{prop}

\subsection{Estimation methods}
The methods of fitting an algebraic curve (or surface) to observed points
can be classified as follows.\eject

\paragraph{Algebraic distance methods,\protect\eatdot}\
where the residuals in the equations for the algebraic curve
are minimized.
For example, the minimum point of the sum of squared residuals
$\sum_{i=1}^n (A x_i^2 + 2 B x_i y_i + C y_i^2 +
2 D x_i + 2 E y_i + F)^2$
(with some normalizing constraint in order to avoid
$A = B = \ldots = F = 0$) in the conic fitting problem
and
$\sum_{i=1}^n (k_1 x_i + h_1 - y_i)^2 (k_2 x_i + h_2 - y_i)^2$
in the two-line fitting problem is called
the \emph{ordinary least squares\/} (OLS) estimator.

The criterion function for the OLS estimator is simple enough
and can be adjusted so that the resulting estimator is consistent
(under some conditions).
Such an estimator is called the
\emph{adjusted least squares\/} (ALS) estimator.
The OLS and ALS estimators are method-of-moments estimators,
meaning that the criterion functions for the estimators
are polynomials whose coefficients are sample moments of
coordinates of the observed points.
Hence, the OLS and ALS estimators can be computed efficiently.

In order to obtain parameters of two lines,
the observed points are fitted with a conic section,
and then the parameters of the conic section
are used to obtain the parameters of two lines.
There are some papers where this idea is used.

The problem of estimating the fundamental matrix
for two-camera view is considered in \cite{Kukush2002}.
The fundamental matrix is a singular matrix
whose left and right null-vectors are
the coordinates of each camera in
 the coordinate system of the other camera.
Initially, the ALS estimator of the fundamental matrix
is evaluated.
Then it is projected so that
the estimated fundamental matrix is singular.

In \cite{Vidal2005} the problem of segmentation of
a finite-dimensional vector space
onto linear subspaces is considered,
and the generalized principal component analysis method
is introduced.
The sample is fitted with an algebraic cone
(a~set of points that satisfy a homogeneous algebraic equation)
by the OLS method.
Then subspaces are extracted from the algebraic cone
with use of a~small learning sample.
An application of segmentation of a vector space
onto hyperplanes for
searching  planes on binocular image is given in \cite{Zelnik-Manor1999}.

In \cite{Waibel2015} an ellipsoid fitting problem with a constraint
 such that a center of the ellipsoid lies
on a given line is considered.
The algebraic distance with embedded constraint is
minimized. The analytical (behavioral) properties
of the optimization problem are studied.
We consider a conic section fitting problem
but with different constraint---the conic is degenerated to a couple of straight lines.

\paragraph{Geometric distance methods,\protect\eatdot}\
where distances between the estimated curve and each point
are minimized.
The sum of squares of those distances is minimized, and
the {\it orthogonal regression} (OR) estimator is obtained.

A numerical algorithm for evaluation of the orthogonal regression
estimator is presented in monograph \cite{Ahn2004}.

The orthogonal regression is consistent
in the single straight line fitting problem
\cite[Section 1.3.2(a)]{ChengVanNess1999}.
In nonlinear models, the estimator may be inconsistent.
There is a one-step correction procedure in explicit
and implicit models
\cite{Fazekas2004,Repetatska2010}
with application in the ellipsoid fitting model \cite{Repetatska2010}.
However, in the two-line fitting model, the correction from
\cite{Repetatska2010} is unstable.

\paragraph{Probabilistic methods}\
They are used to obtain the \emph{maximum likelihood\/} (ML)
estimator and Bayes estimators.

\subsection{Notation}
Let $\{A_n, \; n=1, \, 2,\, \ldots\}$
be a sequence of random events.
The random event $A_n$ is said to hold eventually if
almost surely there exists $n_0$ such that
$A_n$ occurs for all $n \ge n_0$.
In other words, the random event $A_n$ holds eventually
if and only if it does not occur  only  for finitely many $n$ almost surely.

The estimator $\hat\beta$ is called asymptotically normal
if $\sqrt{n} (\hat\beta - \beta_{\rm true}) \to N(0, \varSigma)$
in distribution, were the asymptotic covariance matrix
$\varSigma$ may be singular,
and $n$ is the sample size.
This definition differs from the conventional one
adopted in asymptotic theory because here
only $\sqrt{n}$-asymptotic normality is considered.

Let $\mathbi{\zeta} \sim N(\mathbi{\mu}, \varSigma)$
be a bivariate random vector.
Then $E = \{\mathbi{z} : (\mathbi{z} - \mathbi{\mu})^\top\allowbreak
\varSigma^{-1} (\mathbi{z} - \mathbi{\mu}) \le 1\}$
is called the 40\% ellipsoid of the normal distribution
because $\Prob (\mathbi{\zeta}\in E) \approx 0.3935$.
This is the ellipsoid where the probability density function
is at least $0.3679$ of its maximum.

\subsection{Outline}
In Section~\ref{sec:estimators}, we construct five estimators
for parameters of the two line fitting model.
In Section~\ref{sec:equivariance},
we propose two definitions of the equivariance of an estimator and  state that all of the five estimators
are equivariant.
The estimators are compared numerically in Section~\ref{sec:simulations}.
The proofs are given in Appendix~\ref{Appendix}.

\section{Estimators}\label{sec:estimators}

\subsection{ALS2 estimator and its projections}
The two-line fitting model is a restriction of
the conic section fitting model.
A~couple of lines defined by the equation
$(k_1 \xi - \eta + h_1) (k_2 \xi - \eta + h_2) = 0$
is a degenerate conic section
\begin{equation}
A \xi^2 + 2 B \xi \eta + C \eta^2 + 2 D \xi + 2 E \eta +
F = 0, \label{eq:conicxieta}
\end{equation}
with coefficients
\begin{equation}
\begin{aligned}
 & A = C k_1 k_2,              &  & 2 D = C(k_1 h_2 + k_2 h_1), \\
 & 2 B = -C (k_1 + k_2),\quad  &  & 2 E = -C(h_1 + h_2),        \\
 &                             &  & F = C h_1 h_2,
\end{aligned} \label{eqs:2linesTOconic}
\end{equation}
with a constraint $C \neq 0$.

The conic section ALS2 estimator provides estimation of the error variance~$\sigma^2$ and the coefficients $A,\allowbreak B, \ldots,\allowbreak F$.

Denote by $\nu(i)\in \{1, 2\}$ the indicator of a line which the true point
$(\xi_i,\eta_i)$ belongs to. Equation~\eqref{q11a} can be rewritten as
\[
\eta_i = k_{\nu(i)} \xi_i + h_{\nu(i)},
\quad i=1,2,\ldots.
\]
The indicator $\nu(i)$ is nonrandom in the {functional} model,
and it is a random variable in the {structural} model.
%\textcolor{red}{\{AU query: ``functional'' both times?\}}

\begin{prop}\label{prop:2lcALS2f}
Let, in the functional model,\begin{gather*}
\sum_{n=1}^\infty \frac{\xi_n^6}{n^2} <
\infty;
\\
\mbox{either}\quad k_1 \neq k_2 \quad \mbox{or} \quad
\begin{cases}
h_1 \neq h_2, \\
\sup_{n\ge 1} \frac{1}{n}
\sum_{i=1}^n \xi_i^2 < \infty;\ \mbox{and}
\end{cases}
\\
\liminf_{n\to\infty} \lambda_{\min} \left(
\frac{1}n \sum_{\substack{i=1,\ldots,n\\ \nu(i)=j}} \begin{pmatrix}
1        & \xi_i^{} & \xi_i^2 \\[2pt]
\xi_i^{} & \xi_i^2  & \xi_i^3 \\[2pt]
\xi_i^2  & \xi_i^3  & \xi_i^4
\end{pmatrix} \right) > 0
\quad \mbox{for $j=1,\,2$}.
\end{gather*}
Then the ASL2 estimators $\hatboldbeta$ and $\hat\sigma^2$
are strongly consistent in the sense of~\eqref{eq:consHatbeta1}
and~\eqref{eq:consHatsigma2}.
\end{prop}

There are two cases where the structural model is not identifiable.
If the common distribution of the true points is concentrated
on a straight line and on a single point (presumably not on the line), that is,
\begin{equation}
\label{nonindef2} \exists \,\mbox{line}\, \ell \subset \mathbb{R}^2 \;
\exists z \in\mathbb{R}^2 : \supp{(\xi_1,
\eta_1)} \subset \ell \cup \{z\},
\end{equation}
then there are many ways to fit the true points with two lines.
If the common distribution of the true points is concentrated in
four points, that is,
\begin{equation}
\label{nonindef1} \# \supp{(\xi_1, \eta_1)} = 4,
\end{equation}
then there are three ways to fit the true points with two lines
(unless three of the four points lie on a straight line,
which is a particular case of \eqref{nonindef2}).

\begin{prop}\label{prop:2lcALS2s}
In the structural model, assume that
$\ME |\xi_1|^3 < \infty$
and that nonidentifiability conditions \eqref{nonindef2}
and \eqref{nonindef1} do not hold.
Then the ALS2 estimator is strongly consistent in the sense of
\eqref{eq:consHatbeta1} and \eqref{eq:consHatsigma2}.
\end{prop}

In order to estimate the parameters $k_1$, $h_1$, $k_2$, and $h_2$,
we can solve Eqs.~\eqref{eqs:2linesTOconic}.
With ignoring the last equation $F = C h_1 h_2$,
the solution is
\begin{gather}
\label{AE-to-kh-a} k_{1,2} = \frac{-B \pm \sqrt{B^2 - AC}}{C},
\\
h_1 = \frac{2 (D + k_1 E)} {C (k_2 - k_1)},
\\
h_2 = \frac{2 (D + k_2 E)} {C (k_1 - k_2)}. \label{AE-to-kh-c}
\end{gather}

Substituting the elements of the ALS2 estimator
$\hatboldbeta = (\hat A, 2 \hat B, \hat C, 2\hat D, 2\hat E, \hat F)^\top$
into the right-hand side of \eqref{AE-to-kh-a}--\eqref{AE-to-kh-c},
we obtain an ``ignore-$\widehat F$'' estimator:
\begin{gather}
\hat k_{1,2} = \frac{-\hat B \pm \sqrt{\hat B^2 - \hat A\hat C}}{\hat C}, \label{AE-to-kh-esta}
\\
\hat h_1 = \frac{2 (\hat D + \hat k_1 \hat E)}                {
\hat C (\hat k_2 - \hat k_1)},
\\
\hat h_2 = \frac{2 (\hat D + \hat k_2 \hat E)}                {
\hat C (\hat k_1 - \hat k_2)}. \label{AE-to-kh-estc}
\end{gather}

If the  conic section estimated by the ALS2 estimator is a hyperbola,
then the ``ignore-$\widehat F$'' estimate of the two lines
comprises the asymptotes of the hyperbola.

Choose the sign $\pm$ in \eqref{AE-to-kh-esta}
such that $\hat k_1 < \hat k_2$.

We need the notation
\[
(k_1, h_1, k_2, h_2)^\top
= \betatolines(\boldbeta)
\]
for the function that expresses the line parameters $k_1$, $h_1$, $k_2$, $h_2$
in elements of $\boldbeta$ and is
defined \xch{by}{by Eqs.}~\eqref{AE-to-kh-a}--\eqref{AE-to-kh-c}.
With this notation, we can write
\[
(\hat k_1, \hat h_1, \hat k_2, \hat
h_2)^\top = \betatolines(\hatboldbeta) = \betatolines(
\widetilde{\boldbeta}) .
\]

\begin{prop}\label{prop:ignF_func_cons}
In the functional model, assume the following:
\begin{gather*}
k_1 < k_2,
\\
\sum_{n=1}^\infty \frac{\xi_n^6}{n^2} <
\infty,\quad \mbox{and}
\\
\liminf_{n\to\infty} \lambda_{\min} \left(
\frac{1}n \sum_{\substack{i=1,\ldots,n\\ \nu(i)=j}} \begin{pmatrix}
1        & \xi_i^{} & \xi_i^2 \\[2pt]
\xi_i^{} & \xi_i^2  & \xi_i^3 \\[2pt]
\xi_i^2  & \xi_i^3  & \xi_i^4
\end{pmatrix} \right) > 0
\quad \mbox{for $j=1,\,2$}.
\end{gather*}
Then the ``ignore-$\widehat F$'' estimator of the parameters of two lines
is strongly consistent, that is,
\[
\hat k_j \to k_j, \qquad \hat h_j \to
h_j, \quad j=1,\,2,
\]
 as $n \to \infty$ almost surely.
\end{prop}

\begin{prop}\label{prop:ignF_struct_cons}
If in the structural model, $k_1 < k_2$,
$\ME |\xi_1|^3 < \infty$,
and neither condition \eqref{nonindef2} nor condition \eqref{nonindef1}
holds,
then the ``ignore-$\widehat F$'' estimator is consistent.
\end{prop}

Now, we state the asymptotic normality of the
``ignore-$\widehat F$'' estimator.

\begin{prop}\label{prop:ignF_func_AN}
In the functional model, assume the following:
\begin{itemize}
\item $k_1 < k_2$,
\item for $j=1, 2$ and $p=0, 1, \ldots, 6$, the following limits
exist and are finite:
\[
\mu_p^{(j)} := \lim_{n\to\infty}
\frac{1}n \sum_{\substack{
i=1,\ldots,n \\ \nu(i)=j}} \xi_i^p.
\]
\item for $j=1$ and $j=2$, the matrices
\[
\begin{pmatrix}
\mu_0^{(j)} & \mu_1^{(j)} & \mu_2^{(j)} \\[2pt]
\mu_1^{(j)} & \mu_2^{(j)} & \mu_3^{(j)} \\[2pt]
\mu_2^{(j)} & \mu_3^{(j)} & \mu_4^{(j)}
\end{pmatrix}
\]
are nonsingular.
\end{itemize}
Then the ``ignore-$\widehat{F}$'' estimator
$(\hat k_1, \hat h_1, \hat k_2, \hat h_2)^\top$
is asymptotically normal, namely
\begin{equation}
\label{eq:ano_khkh} \sqrt{n} \begin{pmatrix}
\hat k_1 - k_1 \\
\hat h_1 - h_1 \\
\hat k_2 - k_2 \\
\hat h_2 - h_2
\end{pmatrix} \distrto N\bigl(0, K \varSigma_{\tilde\beta}
K^\top\bigr),
\end{equation}
where $\varSigma_{\tilde\beta}$ is the asymptotic
covariance matrix of $\tilde\beta$, and
$K$ is the $4\times 6$ matrix of derivatives of the mapping
$(A, 2B, C, 2D, 2E, F)^\top \mapsto (k_1, h_1, k_2, h_2)^\top$
 defined in \eqref{AE-to-kh-a}--\eqref{AE-to-kh-c}
at the true parameters $\betanormaltrue$,
that is,
\[
K = \frac{{\rm d} \betatolines(\boldbeta)}{{\rm d}\boldbeta^\top} \bigg\rrvert _{\mbox{\scriptsize{$\boldbeta{=}\betanormaltrue$}}} .
\]
\end{prop}

The matrix $K \varSigma_{\tilde\beta} K^\top$ is nonsingular.

\begin{prop}\label{prop:ignF_struct_AN}
If, in the structural model,
$k_1 < k_2$,
 $\ME \xi_1^6 < \infty$,
and neither \eqref{nonindef2} nor \eqref{nonindef1} holds,
then the ``ignore-$\widehat{F}$'' estimator is asymptotically normal,
that is,  \eqref{eq:ano_khkh} holds.
\end{prop}

\begin{remark}
The estimators $\hat k_1$, $\hat h_1$, $\hat k_2$, and $\hat h_2$ obtained in
\eqref{AE-to-kh-esta}--\eqref{AE-to-kh-estc}
do not change if $\hat A$, $\widehat B$, \ldots, $\widehat E$
are multiplied by a common factor.
So it does not matter which normalization of $\boldbeta$ is used.
\end{remark}

Equation~\eqref{eq:conicxieta} represents a couple of intersecting straight lines
if and only if
\begin{equation}
\begin{vmatrix} A & B & D
\\
B & C & E
\\
D & E & F \end{vmatrix} = 0 \quad \mbox{and} \quad A C < B^2 .
\end{equation}
Denote
\begin{gather*}
\Delta(\boldbeta) = \begin{vmatrix} A & B & D
\\
B & C & E
\\
D & E & F \end{vmatrix} = A C F + 2 B D E - A E^2 - C
D^2 - B^2 F,
\\
\Delta'(\boldbeta) = \frac{{\rm d} \Delta(\boldbeta)}{
{\rm d}\boldbeta^\top} = \bigl(CF{-}E^2,
DE{-}BF, AF{-}D^2, BE{-}CD, BD{-}AE, AC{-}B^2\bigr),
\end{gather*}
where the function $\Delta(\boldbeta)$ and its derivative
$\Delta'(\boldbeta)$ are evaluated at the point
$\boldbeta = (A, 2B, C, 2D, 2E, F)^\top$.

Perform one-step update of the estimator $\widetilde{\boldbeta}$ to make
it closer to the surface\break $\Delta(\boldbeta)=0$:
\begin{equation}
\label{eq:onestepupdate} \widetilde{\boldbeta}_{\rm 1st} = \widetilde{\boldbeta} -
\frac{\Delta(\widetilde{\boldbeta})}{
\Delta'(\widetilde{\boldbeta}) \widehat{\varSigma}_{\tilde\beta}
\Delta'(\widetilde{\boldbeta})^\top} \widehat{\varSigma}_{\tilde\beta} \Delta'(
\widetilde{\boldbeta})^\top.
\end{equation}
Then use expressions \eqref{AE-to-kh-a}--\eqref{AE-to-kh-c}
to estimate $k_1, h_1, k_2, h_2$:
\[
(\hat {k}_{1,{\rm 1st}}, \hat {h}_{1,{\rm 1st}}, \hat {k}_{2,{\rm 1st}}, \hat
{h}_{2,{\rm 1st}})^\top = \betatolines(\widetilde{
\boldbeta}_{\rm 1st}) .
\]

\begin{prop}\label{prop:can_1st}
Under the conditions of Proposition~\ref{prop:ignF_func_AN}
in the functional model or under the conditions of
Proposition~\ref{prop:ignF_struct_AN} in the structural model,
the estimator
$(\hat k_{1,{\rm 1st}},
 \hat h_{1,{\rm 1st}},\break
 \hat k_{2,{\rm 1st}},
 \hat h_{2,{\rm 1st}})^\top$
is consistent and asymptotically normal, and its
asymptotic covariance matrix is equal to
\[
K \biggl( \varSigma_{\tilde\beta} - \frac
{\varSigma_{\tilde\beta} \Delta'(\betanormaltrue)^\top
\Delta'(\betanormaltrue) \varSigma_{\tilde\beta}}{
\Delta'(\betanormaltrue)
\varSigma_{\tilde\beta} \Delta'(\betanormaltrue)^\top} \biggr) K^\top
.
\]
\end{prop}

\begin{remark}
The normalization of the estimator $\hatboldbeta$ affects its asymptotic
covariance matrix, and hence has effect on the estimates
$(\hat k_{1,{\rm 1st}},
 \hat h_{1,{\rm 1st}},
 \hat k_{2,{\rm 1st}},
 \hat h_{2,{\rm 1st}})^\top$.
However, the normalization does not affect the asymptotic covariance matrix
of
$(\hat k_{1,{\rm 1st}},
 \hat h_{1,{\rm 1st}},
 \hat k_{2,{\rm 1st}},\break
 \hat h_{2,{\rm 1st}})^\top$.
\end{remark}

\subsection{Orthogonal regression estimator}
The sum of squared distances between each observed point
and the closer of two lines is equal to
\begin{equation}
\label{eq:ORQf} Q(k_1, h_1, k_2,
h_2) = \sum_{i=1}^n \min
\biggl( \frac{(y_i - k_1 x_i - h_1)^2}{k_1^2 + 1}, \, \frac{(y_i - k_2 x_i - h_2)^2}{k_2^2 + 1} \biggr).
\end{equation}
The \emph{orthogonal regression} estimator is a Borel-measurable function
of observations such that
\[
(\hat k_{1,\rm OR}, \hat h_{1,\rm OR}, \hat k_{2,\rm OR}, \hat
h_{2,\rm OR}) \in \mathop{\argmax}\limits
_{(k_1, h_1, k_2, h_2) \in \mathbb{R}^4} Q(k_1, h_1,
k_2, h_2) .
\]

In the functional model, the orthogonal regression estimator is
the maximum likelihood estimator. However, because the dimension of
parameter space grows as the sample size is increasing,
the orthogonal regression estimator may be inconsistent.

\subsection{Parametric maximum likelihood estimator}\label{ss:PML}
The estimator is constructed in the structural model,
so it should be called the structural maximum likelihood estimator.

If a Gaussian distribution of a random point $(\xi, \eta)$
is concentrated on a~straight line $\eta = k \xi + h$,
then it is a singular normal distribution:
\begin{equation}
\label{q122a} ( \xi, \eta ) \sim N \left( \begin{pmatrix} \mu_\xi \\ k \mu_\xi + h \end{pmatrix}, \, \begin{pmatrix}
\sigma_\xi^2   & k \sigma_\xi^2 \\
k \sigma_\xi^2 & k^2 \sigma_\xi^2
\end{pmatrix} \right),
\end{equation}
where $\mu_\xi$ and $\sigma_\xi^2$ are the expectation and variance
of the random variable $\xi$.
Note that the covariance matrix
$\sigma^2_\xi  \big(\begin{smallmatrix}
1 & k \\ k & k^2 \end{smallmatrix} \big)$
is singular and positive semidefinite.

If the distribution of a random point $(\xi_i, \eta_i)$
is concentrated on two straight
lines $\eta = k_1 \xi + h_1$ and $\eta = k_2 \xi + h_2$
and the distribution on each line is Gaussian, then, due to
\eqref{q122a}, the conditional distributions are
\begin{equation*}
%\label{q122a}
 \bigl[ \begin{pmatrix} \xi_i, \eta_i \end{pmatrix} \; \mid \; \! \nu(i) = j \bigr] \sim N \left( \begin{pmatrix} \mu_{j\xi} \\ k_j \mu_{j\xi} + h_j \end{pmatrix},
\, \begin{pmatrix}
\sigma_{j\xi}^2   & k_j \sigma_{j\xi}^2 \\
k_j \sigma_{j\xi}^2 & k^2_j \sigma_{j\xi}^2
\end{pmatrix} \right) = N(\boldmu_j, \varSigma_{0j})
\end{equation*}
for $j=1,\,2$.
The matrices $\varSigma_{0j}$ are positive semidefinite and singular,
that is,\break $\lambda_{\min}(\varSigma_{01}) = \lambda_{\min}(\varSigma_{02}) = 0$,
and the points $\boldmu_j$ are
the centers of Gaussian distribution of the points on each line.

The distribution of $(\xi_i,\eta_i)$
is a mixture of two singular normal distributions
\begin{equation}
\label{q122b} \begin{pmatrix} \xi_i \\ \eta_i \end{pmatrix} \sim \mbox{mixture of} \begin{cases}
N(\boldmu_1, \varSigma_{01}) & \mbox{with weight $p$}, \\
N(\boldmu_2, \varSigma_{02}) & \mbox{with weight $1-p$}, \\
\end{cases}
\end{equation}
where $p = \Prob(\nu(i)=1) = \Prob(\nu(1)=1)$ is the probability
that the point $(\xi_i,\eta_i)$ lies of the first line.

The distribution of the observed points is also a mixture of
two Gaussian distributions
\begin{equation}
\label{q122c} \begin{pmatrix} x_i \\ y_i \end{pmatrix} \sim \mbox{mixture of} \begin{cases}
N(\boldmu_1, \varSigma_1) & \mbox{with weight $p$}, \\
N(\boldmu_2, \varSigma_2) & \mbox{with weight $1-p$}, \\
\end{cases}
\end{equation}
with $\varSigma_j = \varSigma_{0j} + \sigma^2 {\bf I}$,
where $\sigma^2$ is the error variance; see \eqref{q11c}.
Note that $\lambda_{\min} (\varSigma_1) =
\lambda_{\min} (\varSigma_2) = \sigma^2$.

The likelihood function for the sample of points with
a mixture of two normal distributions is
\begin{equation}
\label{eq:PMLLf} L(p, \boldmu_1, \varSigma_1,
\boldmu_2, \varSigma_2) = \prod
_{i=1}^n \bigl(p \phi_{N(\mu_1,\varSigma_1)}(x_i,y_i)
+ (1{-}p) \phi_{N(\mu_1,\varSigma_1)}(x_i,y_i)\bigr),
\end{equation}
where
\[
\phi_{N(\mu,\varSigma)}(x,y) = \frac{1}{2 \pi \sqrt{\det \varSigma}} \exp \biggl\{ -\frac{1}2
\left( \begin{pmatrix} x\\y \end{pmatrix} - \mu \right)^\top \varSigma^{-1} \left( \begin{pmatrix} x\\y \end{pmatrix} - \mu \right)
\biggr\}
\]
is the density of a bivariate normal distribution.

One method of evaluating the maximum likelihood estimator
is as follows:
\begin{enumerate}
\item
Find the point of conditional minimum
\[
(\widehat{\boldmu}_1, \widehat{\varSigma}_1, \widehat{
\boldmu}_2, \widehat{\varSigma}_2) = \mathop{\argmin}
\limits_{\substack{
 \mu_1, \varSigma_1, \mu_2, \varSigma_2\\
 {\rm such\ that\ } \lambda_{\min} (\varSigma_1) = \lambda_{\min} (\varSigma_2)}}
\min_{p\in[0, 1]} L(p, \boldmu_1, \varSigma_1,
\boldmu_2, \varSigma_2) .
\]

\item
Set
\begin{equation}
\label{eq:PMLhvarer} \hat\sigma^2 = \lambda_{\min} (
\varSigma_1) = \frac{1}2 \bigl( \hat\sigma_{1xx} +
\hat\sigma_{1yy} - \sqrt{(\hat\sigma_{1xx}-\hat
\sigma_{1yy})^2 + 4\hat\sigma^2_{1xy}}\,
\bigr).
\end{equation}
Here $\hat\sigma_{jxx}$, $\hat\sigma_{jxy}$, and $\hat\sigma_{jyy}$ are
the entries of the matrix $\widehat\varSigma_j$,
and $\hat\mu_{jx}$ and $\hat\mu_{jy}$ are the elements of the vector
$\widehat{\boldmu}_j$:
\[
\widehat{\boldmu}_j = \begin{pmatrix}
\hat\mu_{jx} \\ \hat\mu_{jy}
\end{pmatrix}, \qquad \widehat{
\varSigma}_j = \begin{pmatrix}
\hat\sigma_{jxx} & \hat\sigma_{jxy} \\
\hat\sigma_{jxy} & \hat\sigma_{jyy}
\end{pmatrix}.
\]

\item
Find the estimates $\hat k_1$, $\hat h_1$, $\hat k_2$,
$\hat h_2$ from the equations
\[
\widehat{\boldmu}_j = \begin{pmatrix}
\mu_{jx} \\
\hat k_j \mu_{jx} + \hat h_j
\end{pmatrix}, \qquad \widehat{
\varSigma}_j = \begin{pmatrix}
\hat\sigma_{j\xi}^2 + \hat\sigma^2 & \hat k_j \hat\sigma_{j\xi}^2 \\
\hat k_j \hat\sigma_{j\xi}^2 & \hat k_j^2 \hat\sigma_{j\xi}^2 + \hat\sigma^2
\end{pmatrix},
\]
that is, set
\begin{gather}
\label{eq:PMLhkj} \hat k_j = \frac{\hat\sigma_{jxy}}{\hat\sigma_{jxx}-\hat\sigma^2},
\\
\label{eq:PMLhhj} \hat h_j = \hat \mu_{jy} - \hat
k_j \hat \mu_{jx} .
\end{gather}
\end{enumerate}

The denominator $\hat\sigma_{jxx} - \hat\sigma^2$ may be equal to 0
with some positive probability. Occurrence of this event  means
that the estimated figure is a straight line and a single point
outside the line rather than two straight lines.

In order to make the statement of consistency easier,
assume that $k_1 < k_2$ and choose the estimator such that
$\hat k_1 \le \hat k_2$.

\subsection{RBAN moment estimator}\label{ss:RBANmome}
The \emph{regular best asymptotically normal} (RBAN) estimators
were developed by Chiang \cite{Chiang1956}.
Our RBAN moment estimator differs from the original RBAN
so that not only the observed points $(x_i, y_i)$,
but also monomials $x_i^p y_i^q$, $p+q\le 4$, are averaged.

Introduce the 14-dimensional vectors whose elements are
the monomials of coordinates of observed points:
\begin{align}
\label{eq:RMANmxy}m(x,y) &= \bigl(x^4, x^3 y, x^2 y^2, x y^3, y^4, x^3, x^2 y, x y^2, y^3,  x^2, x y, y^2, x, y\bigr)^\top,\\
\nonumber \mbox{\boldmath$m$}_i &= m(x_i, y_i).
\end{align}

Evaluate the average and sample covariance matrix
of the vectors $\mbox{\boldmath$m$}_i$:
\[
\overline{\mbox{\boldmath$m$}} = \frac{1}n \sum
_{i=1}^n \mbox{\boldmath$m$}_i, \qquad
\varSigma_m = \frac{1}n \sum_{i=1}^n
(\mbox{\boldmath$m$}_i - \overline{\mbox{\boldmath$m$}}) (\mbox{
\boldmath$m$}_i - \overline{\mbox{\boldmath$m$}})^\top.
\]

Denote
\[
f_1\bigl(k, h, \sigma^2; (\mu_p)_{q=1}^4
\bigr) = \ME m(\xi+\delta, k \xi + h + \varepsilon),
\]
where $\xi$,  $\delta$, and $\varepsilon$
are independent random variables such that
\[
\ME \xi^q = \mu_q \quad \mbox{and} \quad (\delta,\varepsilon)^\top {\sim}\break N(0, \sigma^2 {\bf I}).
\]
Basically, the function $f_1$ is defined for all
$\mu_p$, $p=1,\ldots,4$,
that comprise possible 4-tuples of moments of a random variable,
that is,  satisfy
\begin{align*}
&\qquad \qquad \qquad \qquad \qquad \qquad \qquad \mu_2 - \mu_1^2 \ge 0,&\\
&\bigl(\mu_4 \,{-}\, 4 \mu_3 \mu_1 \,{+}\, 6\mu_2 \mu_1^2 \,{-}\, 3 \mu_1^4\bigr) \bigl(\mu_2 \,{-}\, \mu_1^2\bigr) \,{-}\, \bigl(\mu_3 \,{-}\, 3 \mu_2 \mu_1 \,{+}\, 2\mu_1^3\bigr)^2 \,{-}\, \bigl(\mu_2 {-}\,\mu_1^2\bigr)^3 \,{\ge}\, 0;&
\end{align*}
see \cite{Rohatgi1989}\@.
However, since the elements of the vector-function $f_1$
are polynomials of its arguments,
it can be extended to $\mathbb{R}^{7}$.

Denote
\begin{align*}
&f_2 \bigl(k_1, h_1, k_2,
h_2, \sigma^2; p, \bigl(\mu_q^{(j)}
\bigr)_{j=1,q=1}^{2\mskip28mu 4} \bigr)
\\
&\quad = p f_1 \bigl(k_1, h_1,
\sigma^2; \bigl(\mu_q^{(1)}\bigr)_{q=1}^4
\bigr) + (1{-}p) f_1 \bigl(k_2, h_2,
\sigma^2; \bigl(\mu_q^{(2)}\bigr)_{q=1}^4
\bigr).
\end{align*}

In the structural model,
\[
\ME \mbox{\boldmath{$m$}}_i = f_2 \bigl(k_1,
h_1, k_2, h_2, \sigma^2; \Prob
\bigl(\nu(1){=}1\bigr), \bigl(\ME\bigl[\xi_1^q \mid
\nu(1){=}j\bigr] \bigr)_{j=1,q=1}^{2\mskip28mu 4} \bigr).
\]

Consider the equation
\begin{equation}
\label{eq:momest1} f_2 \bigl(\hat k_1, \hat
h_1, \hat k_2, \hat h_2, \hat
\sigma^2; \hat p, \bigl(\hat\mu_q^{(j)}
\bigr)_{j=1,q=1}^{2\mskip28mu 4} \bigr) = \overline{\mbox{\boldmath$m$}} .
\end{equation}
It is a system of 14 equations in 14 variables.
\xch{If}{If Eq.}~\eqref{eq:momest1} has a solution,
then the moment estimator can be defined as one of the solutions.
\xch{However,}{However, Eq.}~\eqref{eq:momest1} may have no solution.

In the rest of Section~\ref{ss:RBANmome},
$\mu_{\bullet}^{(\bullet)} = (\mu_{q}^{(j)})_{j=1,q=1}^{2\mskip28mu 4}$
is a $2\times 4$ matrix.

The estimator is defined as a point where
$(f_2(\ldots) - \overline{\mbox{\boldmath$m$}})^\top \varSigma_m^{-1}
 (f_2(\ldots) - \overline{\mbox{\boldmath$m$}})$
 attains its minimum:
\begin{align*}
\bigl(&\hat k_1, \hat h_1, \hat k_2, \hat
h_2, \hat\sigma^2\bigr)
\\
&= \mathop{\argmin}\limits
_{k_1,\ldots,\sigma^2} \min_{p,\mu_{\bullet}^{(\bullet)}} \bigl(f_2
\bigl(k_1,\ldots,\sigma^2; p, \mu_{\bullet}^{(\bullet)}
\bigr) \,{-}\, \overline{\mbox{\boldmath$m$}}\bigr)^\top \varSigma_m^{-1}
\bigl(f_2\bigl(k_1,\ldots,\sigma^2; p,
\mu_{\bullet}^{(\bullet)}\bigr) \,{-}\, \overline{\mbox{\boldmath$m$}}\bigr) .
\end{align*}
This minimization problem is similar to that
in Theorem 6 in \cite{Chiang1956}.
The minimum
\[
\min_{p \in \mathbb{R}} \min_{\mu_{\bullet}^{(\bullet)} \in \mathbb{R}^{2\times 4}} \bigl(f_2
\bigl(k_1,\ldots,\sigma^2; p, \mu_{\bullet}^{(\bullet)}
\bigr) - \overline{\mbox{\boldmath$m$}}\bigr)^\top \varSigma_m^{-1}
\bigl(f_2\bigl(k_1,\ldots,\sigma^2; p,
\mu_{\bullet}^{(\bullet)}\bigr) - \overline{\mbox{\boldmath$m$}}
\bigr)^\top
\]
can be evaluated explicitly, and this allows us to reduce the dimension
of minimization problem.
The reduction of dimension of the optimization problem was used, for example,
in \cite{Markovsky2004STLS}.

The routines evaluating the RBAN-moment estimator and
the estimator for its covariance matrix
are developed without rigid theoretical basis;
see Section~\ref{ss:RBAN-momentEval}.

\section{Equivariance}\label{sec:equivariance}
\subsection{Two definitions of equivariance}
The similarity transformation of $\mathbb{R}^2$ is
\begin{equation}
g(\mathbi{z}) = K \mathbf{U} \mathbi{z} + \Delta\mathbi{z}, \quad \mathbi{z}
\in \mathbb{R}^2, \label{eq:equivtransform}
\end{equation}
where $\mathbf{U}$ is an orthogonal matrix,
$K \neq 0$ is a scaling coefficient, and
$\Delta\mathbi{z} \in \mathbb{R}^2$ is an intercept.

The transformation of a sample of points acting elementwise
is also denoted $g(Z)$:
if $Z = \{\mathbi{z}_i, \; i=1,\ldots,n\}$, then
$g(Z) = \{g(\mathbi{z}_i), \; i=1,\ldots,n\}$.

Hereafter, we use  vector notation:
the observed points are denoted $\mathbi{z}_i = (x_i, y_i)^\top$,
and the true points are denoted $\mathbi{\zeta}_i = (\xi_i, \eta_i)^\top$.

The underlying statistical structure is
$(\mathbb{R}^{n\times 2}, \mathcal{B}(\mathbb{R}^{n\times 2}),
P_{Z\mid \theta}, \theta \in \varTheta)$,
where $Z \in \mathbb{R}^{n\times 2}$ is the observed sample,
$Z = \{\mathbi{z}_i, \; i=1,\ldots,n\}$,
$\mathcal{B}(\mathbb{R}^{n\times 2})$ is the Borel $\sigma$-field,
and $\theta$ is a parameter that uniquely identifies the distribution of
the observed points;
$\theta = (\mathbi{\zeta}_1, \ldots, \mathbi{\zeta}_n; \sigma^2)$
in the functional model,
and
$\theta = (P_{\zeta}; \sigma^2)$
in the structural model.
Here $\mathbi{\zeta}_1, \ldots, \mathbi{\zeta}_n$ are points
located on two strait lines,
and $P_{\zeta}$ is a probability measure
concentrated on two straight lines.

The statistical structure is invariant with respect
to transformation
$g$ if the change of the probability measure induced
by the transformation of the sample can be obtained
by some transformation $\tilde g$ of parameters,
that is, if there exists a bijection $\tilde g : \varTheta \to \varTheta$
such that
\[
\forall \theta \in \varTheta : \quad P_{g(Z) \mid \theta} = P_{Z \mid \tilde g(\theta)} .
\]
Here $P_{g(Z) \mid \theta}$ is the induced probability measure;
it is sometimes denoted $P_{g(Z) \mid \theta} = P_{Z \mid \theta} g^{-1}$.

The statistical structure is similarity invariant if
it is invariant with respect to all similarity
transformations of the form~\eqref{eq:equivtransform}.

In order to become similarity invariant,
the underlying statistical structure needs some extension.
We assume that the true points lie on two lines,
which \emph{may\/} be parallel to the $y$-axis.
The following restrictions do not ruin the invariance:
\begin{itemize}
\item The true lines $\ell_1$ and $\ell_2$
intersect each other but do not coincide.
\item The true points $\mathbi{\zeta}_1\ldots,\allowbreak\mathbi{\zeta}_n$
in the functional model or
the set $\supp(P_{\zeta})$ where the true points are concentrated
in the structural model
can be covered with two lines uniquely.
In the structural model, this means that
the nonidentifiability conditions \eqref{nonindef2}
and \eqref{nonindef1} do not hold.
\end{itemize}

With these restrictions, the statistical structure
is invariant with
\[
\tilde g\bigl(\mathbi{\zeta}_1, \ldots, \mathbi{
\zeta}_n; \sigma^2\bigr) = \bigl(g(\mathbi{
\zeta}_1), \ldots, g(\mathbi{\zeta}_n); K^2
\sigma^2\bigr)
\]
in the functional model and
\[
\tilde g\bigl(P_{\zeta}; \sigma^2\bigr) =
\bigl(P_{g(\zeta)}; K^2 \sigma^2\bigr)
\]
in the structural model.

Let $\ellell(\theta)$ and $\sigmatwo(\theta)$
be functions that extract
the parameters of interest.
If, in the functional model,
$\theta=(\mathbi{\zeta}_1, \ldots, \mathbi{\zeta}_n; \sigma^2)$
and points $\xch{\mathbi{\zeta}_1}{{\zeta}_1}, \ldots,\allowbreak  \mathbi{\zeta}_n$
lie on the lines $\ell_1$ and $\ell_2$
or if, in the structural model,
$\theta=(P_{\zeta}; \sigma^2)$
and the probability measure is concentrated
on the union of two lines $\ell_1 \cup \ell_2$,
then $\ellell(\theta) = \{\ell_1, \ell_2\}$.
If $\theta = (\ldots; \sigma^2)$, then $\sigmatwo(\theta) = \sigma^2$.

We treat $\{\ell_1, \ell_2\}$ as an unordered couple,
that is, $\{\ell_1, \ell_2\} = \{\ell_2, \ell_1\}$.

The transformation of the lines parameters and
the transformation of $\sigma^2$ do not
interfere each other, and these transformations
are not interfered
by a~particular location or distribution
of true points on the lines, that is,
the parameter transformation $\tilde g$
is such that there exists transformations
$\tilde g_{\ellell}$ and $\tilde g_{\sigma^2}$
such that
\begin{gather*}
\ellell\bigl( \tilde g (\theta)\bigr) = \tilde g_{\ellell} \bigl(\ellell(
\theta)\bigr),
\\
\sigmatwo\bigl( \tilde g (\theta)\bigr) = \tilde g_{\sigma^2} \bigl(
\sigmatwo(\theta)\bigr).
\end{gather*}
These transformations are
\begin{gather}
\label{eq:fungellell} \tilde g_{\ellell} \bigl(\{\ell_1,
\ell_2\}\bigr) = \bigl\{ g(\ell_1), g(\ell_2)
\bigr\},
\\
\nonumber
\tilde g_{\sigma^2} \bigl(\sigma^2\bigr) =
K^2 \sigma^2.
\end{gather}

The estimator is called equivariant with respect to the transformation $g$
if, when the data are transformed, the estimator follows the inducing
transformation of parameters. The estimator $\hatellell(Z)$
for two lines and the estimator $\hat\sigma^2(Z)$
for error variance are equivariant
with respect to similarity transformation $g$ if
\begin{gather}
\label{eq:defevll} \hatellell\bigl(g(Z)\bigr) = \tilde g_{\ellell} \bigl(
\hatellell(Z)\bigr),
\\
\label{eq:defevs} \hat\sigma^2\bigl(g(Z)\bigr) = \tilde
g_{\sigma^2} \bigl(\hat\sigma^2(Z)\bigr) .
\end{gather}
The estimator is called similarity equivariant if it is
equivariant with
respect to any similarity transformation $g$.

In a fitting problem, an estimator for a ``true figure''
is called \emph{fitting equivariant\/} with respect to
transformation $g(\mathbi{z})$, $g: \mathbb{R}\to\mathbb{R}$ if,
when the sample is transformed, the estimated ``true figure''
follows the same transformation $g$.
An estimator is called \emph{similarity fitting equivariant}
if it is \emph{fitting equivariant\/} with respect to any
similarity transformation.

In the two-line fitting problem, denote
 by $\cup\{\ell_1, \ell_2\} = \ell_1 \cup \ell_2$
the union of a pair of two lines.
An estimator $\hatellell(Z)$ is
similarity fitting equivariant if and only if
for any similarity transformation $g(\mathbi{z})$,
\begin{equation}
\label{eq:deffevll} \cup \hatellell\bigl(g(Z)\bigr) = g\bigl(\cup \hatellell(Z)\bigr)
.
\end{equation}

The similarity fitting equivariant estimator depends on geometry of the plane
and does not depend on the Cartesian coordinate system used.

Because of \eqref{eq:fungellell}, in the two-line
fitting model, the estimator for two lines
$\hatellell(Z)$ is similarity equivariant
if and only if it is similarity fitting equivariant.

\subsection{Similarity equivariance of the five estimators}
Some troubles, which may arise during estimation, are not addressed yet.
\begin{itemize}
\item
The estimation may fail with small positive probability.
For example, the conic section estimated with the ALS2 estimator
is an ellipse with some positive probability,
and if it is, then the ``ignore-$\widehat{F}$'' estimator fails.
(If the estimator is consistent, then
 the failure probability tends to 0 as $n\to\infty$).
\item
The estimation may fail, for example, because the estimated
line should be parallel to the $y$-axis,
but the estimating procedure does not handle such case.
\item
The optimization problem may have multiple extremal points.
For the ALS2 estimator, it may occur that
$\dim\{ \beta : \varPsi_n(\hat\sigma^2) \beta = 0\} > 1$.
\end{itemize}

In order to define the equivariance of an unreliable estimator,
we allow that the estimators fail simultaneously
in both sides \xch{of}{of Eqs.}~\eqref{eq:defevll}, \eqref{eq:defevs},
or \eqref{eq:deffevll}.
Also, we allow that for fixed similarity transformation $g(\mathbi{z})$,
\xch{equation}{Eq.}~\eqref{eq:defevll}, \eqref{eq:defevs},
or \eqref{eq:deffevll} does not hold with probability 0.

The equivariance of the ALS2 estimator in the conic section fitting problem
is verified in \cite[Section 5.5]{Shklyar2007}
(see Theorem~30 there for similarity fitting equivariance).
That implies the equivariance of the ``ignore-$\widehat{F}$'' estimator.

In order to make the \emph{updated before ignore-$\widehat{F}$ step}
estimator equivariant,
we use normalization of the ALS2 estimator \eqref{eq:defbetatilde}
rather than $\|\hat{\mathbi{\beta}}\| = 1$.

The \emph{orthogonal regression\/} estimator and the
parametric \emph{maximum likelihood\/} estimator
are maximum likelihood estimators, but in different models.
Thus, they are equivariant.

The criterion function for the \emph{RBAN-moment\/} estimator
is similarity invariant.
This means that the criterion function does not change when
the data sample follows a similarity transformation and
the parameters follow the inducing transformation.
Thus, the RBAN-moment estimator is equivariant.

\subsection{An example of equivariant but
not fitting equivariant estimator}\label{ss:OddExample}
Consider a further restriction of the mixture-of-two-normal-distributions
model from Section~\ref{ss:PML}. Assume that covariance matrices
of $\varSigma_1$ and $\varSigma_2$ have the same diagonal entries
but additive inverse off-diagonal entries:
\[
\varSigma_1 = \begin{pmatrix}
\sigma_\xi^2 + \sigma^2 & -k \sigma_\xi^2 \\
-k \sigma_\xi^2 & k^2 \sigma_\xi^2 + \sigma^2
\end{pmatrix} \quad \mbox{and} \quad \varSigma_2
= \begin{pmatrix}
\sigma_\xi^2 + \sigma^2 & k \sigma_\xi^2 \\
k \sigma_\xi^2 & k^2 \sigma_\xi^2 + \sigma^2
\end{pmatrix} .
\]

The statistical structure is invariant in scaling of  the $y$-coordinate,
$(x_{\rm new}, y_{\rm new}) = (x_{\rm old},\: r y_{\rm old})$,
$r > 0$.
This transformation maps the lines $y = -k x + h_1$ and
$y = k x + h_2$ onto the lines
$y = -r k x + r h_1$ and $y = r k x + r h_2$, respectively.
The maximum likelihood estimator in this model is
equivariant. However, this equivariance is somewhat
strange.  The transformation of parameters that induces
the scaling of the  $y$-coordinate of the observed points
does not induce the same transformation of the true
points nor the same mapping of the true lines.
The estimated lines follow the transformation of
parameters rather than the transformation of
observed points. This is illustrated in Fig.~\ref{fig:f1}.

\begin{figure}[bt]
\includegraphics[scale=0.95]{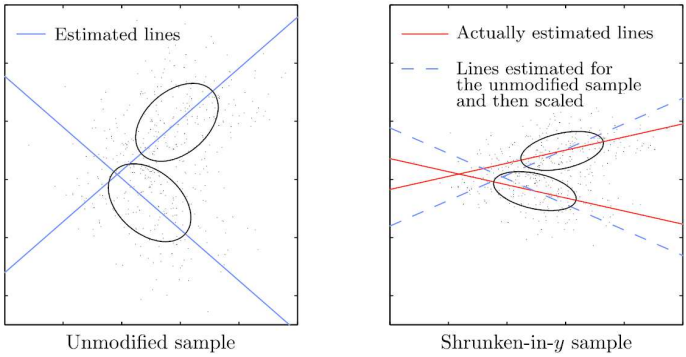}
\caption{Two samples, one of points $(x,y_{\rm old})$
(the unmodified sample) and one of points
$(x,y_{\rm new}) = (x, \frac{1}2 y_{\rm old})$
(the shrunken-in-$y$ sample),
are fitted with two lines
(the estimated lines are the solid lines
on the figures).
The estimated lines for the unmodified sample
(blue solid lines on the left figure)
when scaled with the same transformation as the
observed points are scaled
(blue dashed line on the right figure)
do not coincide with the actually
estimated lines for the shrunken-in-$y$ sample
(red solid lines on the right figure).
The ellipsoids are the 40\% ellipsoids
of the estimated normal distributions
(the compound distributions of the estimated
mixture \xch{distributions) (color figure online)}{distributions)}}\label{fig:f1}
%\vspace*{-3pt}
\end{figure}

Let $k_{\rm old}$ be the true value of the parameter
$k$ before the transformation.
Then after the transformation, the value of the parameter
is
\[
k_{\rm new} = \frac{t
+ \sqrt{
t^2 + 4 r^2 k^2_{\rm old}
\sigma^4_\xiold}}{
2 r k_{\rm old}^{} \sigma^2_\xiold}
\]
with $t = (r^2 k_{\rm old}^2 - 1)\, \sigma^2_\xiold
 + (r^2 - 1)\, \sigma^2_{\rm old}$.
If $0 < r^2 \neq 1$, $k_{\rm old} \neq 0$, and $\sigma^2 > 0$,
then $k_{\rm new} \neq r k_{\rm old}$.
Hence, the maximum likelihood estimator is not fitting equivariant
with respect to scaling of the $y$-coordinate here.

\vspace*{-3pt}\section{Simulations}\label{sec:simulations}\vspace*{-2pt}
\subsection{Simulation setup}\label{ss:SimulationSetup}
A sample of the true points $(\xi_i, \eta_i)$,
$i=1,\ldots,n$,
is generated from a random distribution concentrated on (a subset of) two lines.
Three distributions of the true points are used; see
Fig.~\ref{fig:3distrs}:\vadjust{\eject}
\begin{itemize}
\item a mixture of two singular normal distributions,\vspace*{-3pt}
\item a discrete distribution,\vspace*{-3pt}
\item a uniform distribution on two line segments.
\end{itemize}
\begin{figure}[t!]
\includegraphics[scale=0.95]{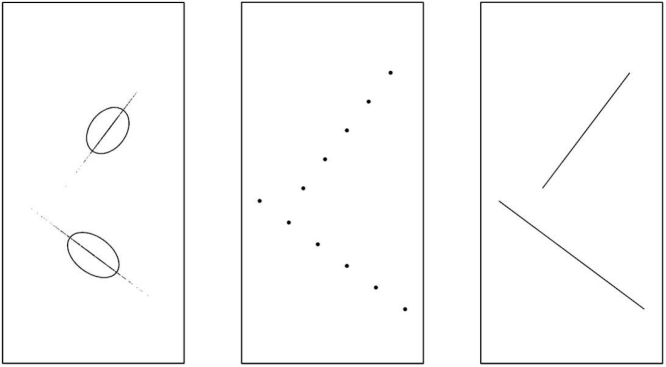}
\caption{Three distributions of the true points:
a mixture of two singular normal distributions,
a discrete distribution, and a uniform distribution
on two line segments. For the first case,
a~sample of 1000 points is plotted,
whereas for the second and third cases,
the support of the distribution of the true
points is plotted.
For the first case, the distribution of the
\emph{observed} points is a mixture of normal
distributions, and 40\% ellipsoids for its
components are plotted}\label{fig:3distrs}\vspace*{-9pt}
\end{figure}

These three distributions of true points are
concentrated on the same two lines
\[
4y=1-3x \quad \mbox{and} \quad 12y=16x+5,
\]
which intersect one another at the point $({-}0.08, 0.31)$.

For the same sample of true points
$\{(\xi_i, \eta_i),\allowbreak\; i=1,\ldots,n\}$,
100 samples of the measurement errors
$\{(\delta_i, \varepsilon_i),\allowbreak\; i=1,\ldots,n\}$,
$(\delta_i, \varepsilon_i)^\top \sim N(0,\sigma^2)$,
are simulated,
and 100 samples of the observed points $(x_i,y_i)$ are obtained;
see~\eqref{q11b1} and~\eqref{q11b2}.
For each sample of the observed points, the estimates
of the parameters of the true lines were evaluated
with the following five methods: two ALS2-based estimators
(the ignore-$\widehat F$ estimator and the estimator
with one-step update of the ALS2 estimator
before the ignore-$\widehat F$ step),
the orthogonal regression estimator,
the parametric maximum likelihood estimator,
and the RBAN moment estimator.

For each estimated couple of lines, the point of their intersection
is found. The 100 estimates of intersection points are averaged, and
their sample standard deviations are evaluated.
For the ALS2-based estimators and the RBAN moment estimator,
the standard errors of the estimators are also evaluated.

\subsection{Notes on computation of particular estimators}
For computation of the \emph{orthogonal regression\/} estimator,
the $k$-means method is used.
Initially, two lines were chosen randomly. Then \emph{classification\/}
and \emph{mean\/} steps are alternated.
On the classification step, the observed points
are split into two clusters based on which line is closer to the point.
(The first cluster contains all the observed points that are
closer to the first line than to the second line,
and the second cluster contains the other observed
points.)
On the means step, each cluster is fitted with
a straight line by the orthogonal regression method
(the two lines are updated).
The algorithm is completed when the classification
step does now change the clusters.
The obtained parameters of the two lines deliver a local minimum
to the criterion function $Q(k_1, h_1, k_2, h_2)$
\eqref{eq:ORQf}.
Trying to obtain the global minimum,
the algorithm is restarted several times
with different initial two lines.

For computation of the \emph{parametric maximum
likelihood estimator},
the ex\-pect\-ation--maximization algorithm
\cite{bilmes1998gentle} is used.
The equation for optimization problem of finding
a minimum of the likelihood function
$L(p, \mathbi{\mu}_1, \varSigma_1, \mathbi{\mu}_2,\varSigma_2)$
\eqref{eq:PMLLf} such that
$\lambda_{\min}(\varSigma_1) = \lambda_{\min}(\varSigma_2)$
is
\begin{gather*}
\sum_{i=1}^n \frac{\phi_{N(\mu_1,\varSigma_1)}(x_i, y_i) -
\phi_{N(\mu_2,\varSigma_2)}(x_i, y_i)}{
p \phi_{N(\mu_1,\varSigma_1)}(x_i, y_i) +
(1{-}p) \phi_{N(\mu_2,\varSigma_2)}(x_i, y_i)} = 0,
\\
\sum_{i=1}^n \frac{p \frac{\partial \phi_{N(\mu_1,\varSigma_1)}(x_i, y_i) }
{\partial \textit{par}} + (1{-}p)
\frac{\partial \phi_{N(\mu_2,\varSigma_2)}(x_i, y_i)}
{\partial \textit{par}}}{
p \phi_{N(\mu_1,\varSigma_1)}(x_i, y_i) +
(1{-}p) \phi_{N(\mu_2,\varSigma_2)}(x_i, y_i)} = 0,
\end{gather*}
where $\textit{par}$ is a vector parameterization of
$(\mathbi{\mu}_1,\varSigma_1,\mathbi{\mu}_2,\varSigma_2)$
such that
$\lambda_{\min}(\varSigma_1) = \lambda_{\min}(\varSigma_2)$.
Hence, the maximum likelihood estimator is a stationary point
of the function
\[
Q_{\bf w}(p, \mathrm{par}) = \sum_{i=1}^n
\bigl( w_i \ln \bigl(p \phi_{N(\mu_1,\varSigma_1)}(x_i,
y_i)\bigr) + (1{-}w_i) \ln \bigl((1{-}p)
\phi_{N(\mu_1,\varSigma_1)}(x_i, y_i)\bigr) \bigr)
\]
with fixed
\[
w_i = \frac {\hat p \phi_{N(\hat\mu_1,\widehat\varSigma_1)}(x_i, y_i)}{
\hat p \phi_{N(\hat\mu_1,\widehat\varSigma_1)}(x_i, y_i) +
(1{-}\hat p) \phi_{N(\hat\mu_2,\widehat\varSigma_2)}(x_i, y_i)} .
\]

The EM algorithm is iterative. Once the $(m-1)$th
approximation
$(p^{(m-1)},\allowbreak \mu_1^{(m-1)},
\allowbreak \varSigma_1^{(m-1)},
\mu_2^{(m-1)},\allowbreak \varSigma_2^{(m-1)})$ is obtained,
the weights are evaluated:
\[
w_i^{(m-1)} = \frac
{p^{(m-1)} \phi_{N(\mu_1^{(m-1)}, \varSigma_1^{(m-1)})} (x_i, y_i)}{
p^{(m-1)} \phi_{N(\mu_1^{(m-1)}, \varSigma_1^{(m-1)})} (x_i, y_i) +
(1{-}p^{(m-1)}) \phi_{N(\mu_2^{(m-1)}, \varSigma_2^{(m-1)})} (x_i, y_i)} .
\]
Then $m$th approximation
$(p^{(m)}, \mu_1^{(m)}, \varSigma_1^{(m)},
\mu_2^{(m)}, \varSigma_2^{(m)})$
is obtained by minimizing
\[
\sum_{i=1}^n \bigl( w_i^{(m-1)}
\ln\bigl(p \phi_{N(\mu_1, \varSigma_1)} (x_i, y_i)\bigr) +
\bigl(1{-}w_i^{(m-1)}\bigr) \ln \bigl( (1{-}p)
\phi_{N(\mu_2, \varSigma_2)} (x_i, y_i)\bigr) \bigr)
\]
under the constraint
$\lambda_{\min}(\varSigma_1) =  \lambda_{\min}(\varSigma_2)$.
The point where the minimum is attained can be explicitly
expressed in $w_i^{(m-1)}$, $x_i$, and $y_i$, $i=1,\ldots,n$.

\subsection{RBAN-moment estimator}\label{ss:RBAN-momentEval}
In case the criterion function
\begin{align*}
Q(\theta) &= Q\bigl(k_1, h_1, k_2,
h_2, \sigma^2\bigr)
\\
&= \min_{p \in \mathbb{R}} \min_{{\bf M} \in \mathbb{R}^{2\times 4}}
\bigl(f_2(\theta; p, {\bf M}) - \mathbi{\overline{m}}
\bigr)^\top \varSigma_m^{-1}
\bigl(f_2(\theta; p, {\bf M}) - \overline{\mbox{\boldmath$m$}}\bigr)
\end{align*}
has multiple minima,
a consistent estimator---that is, the ``ignore-$\widehat{F}$'' estimator---is used as the initial point,
and the criterion function $Q(\theta)$
is searched for a~local minimum nearby.
Here $\theta = (k_1, h_1, k_2, h_2, \sigma^2)^\top$
is a vector meaning the parameters of interest.

The knowledge or misspecification of the parameter $p$
does not affect the estimator for the parameters of interest
$k_1$,  \ldots, $\sigma^2$.
Thus, for estimation of the asymptotic covariance matrix,
assume $p=0.5$ to be known.
The estimator of the asymptotic covariance matrix of $(\theta, {\bf M})$
is
\[
\widehat\varSigma_{\theta,{\bf M}} = \bigl(f'_2(\hat
\theta;\, 0.5, \widehat{\bf M})^\top \varSigma_m^{-1}
f'_2(\hat\theta;\, 0.5, \widehat{\bf M})
\bigr)^{-1},
\]
where
\[
f'_2(\theta; 0.5, {\bf M}) =
\begin{pmatrix}
\frac{\partial f_2(\theta;\, 0.5, {\bf M})}{\partial \theta^\top}, &\hspace*{-9pt} \frac{\partial f_2(\theta;\, 0.5, {\bf M})}{\partial (\vectorize{\bf M})^\top}
\end{pmatrix}.
\]
The estimator of the asymptotic covariance matrix of $\theta$
is the principal submatrix of $\widehat\varSigma_{\theta,{\bf M}}$.

\subsection{Simulation results}\label{ss:simresults}
Average of estimated centers over 100 simulations,
standard deviations over 100 simulations,
and medians of estimated standard errors
are presented in Tables~\ref{table1-Normal}--\ref{table1-Uniform}.
%%%%%%%{tables1.tex}
\begin{table}[t]
\def\arraystretch{1.12}
\caption{Means, standard deviations, and median standard errors
of the estimates of intersection points
for true points having mixture of singular normal distributions}
\label{table1-Normal}\vspace*{-9pt}
\begin{tabular*}{\textwidth}
{@{\extracolsep{\fill}}l@{\qquad\quad}D{.}{.}{4}D{.}{.}{4}%
     @{\qquad\quad}D{.}{.}{4}D{.}{.}{4}%
     @{\qquad\quad}D{.}{.}{4}D{.}{.}{4}@{}}
\hline
\multicolumn{1}{@{}c@{\qquad\quad}}{Method}&
\multicolumn{2}{@{}c@{\qquad\quad}}{Means}&
\multicolumn{2}{@{}c@{\qquad\quad}}{\hbox to 0pt{\hss Standard~deviations\hss}}&
\multicolumn{2}{@{}c@{}}{Standard errors}\\
\hline
True value& -0.08 & 0.31 & & & & \\\hline
\multicolumn{6}{l}{$n=1000, \qquad \sigma=0.1 \quad(\sigma^2 = 0.01)$}\\
\IgnoreF &  -0.1098 & 0.2753 &  0.8611 & 0.3866 &  0.1918 & 0.1125 \\
Update   &  -0.0820 & 0.2912 &  0.0706 & 0.0620 &  0.0437 & 0.0479 \\
OR       &   0.6533 & 3.4524 &  0.0877 & 0.6783 &         &        \\
ML       &  -0.0795 & 0.3077 &  0.0326 & 0.0269 \\
RBAN     &   0.0647 & 0.3759 &  0.3563 & 0.2606 &  0.0350 & 0.0438 \\
\hline
\multicolumn{6}{l}{$n=10000, \qquad \sigma=0.1 \quad(\sigma^2 = 0.01)$}\\
\IgnoreF &  -0.0909 & 0.3052 &  0.0646 & 0.0308 &  0.0601 & 0.0303 \\
Update   &  -0.0796 & 0.3080 &  0.0127 & 0.0156 &  0.0124 & 0.0155 \\
OR       &   0.5492 & 3.1488 &  0.0175 & 0.1701 &         &        \\
ML       &  -0.0776 & 0.3083 &  0.0103 & 0.0088 &         &        \\
RBAN     &  -0.0789 & 0.3100 &  0.0126 & 0.0154 &  0.0127 & 0.0154 \\
\hline
\multicolumn{6}{l}{$n=100000, \qquad \sigma=0.1$}\\
\IgnoreF &  -0.0799 & 0.3101 &  0.0211 & 0.0095 &  0.0188 & 0.0093 \\
Update   &  -0.0801 & 0.3098 &  0.0037 & 0.0042 &  0.0039 & 0.0047 \\
OR       &   0.5606 & 3.2041 &  0.0063 & 0.0484 &         &        \\
ML       &  -0.0801 & 0.3101 &  0.0030 & 0.0025 &         &        \\
RBAN     &  -0.0801 & 0.3101 &  0.0038 & 0.0042 &  0.0039 & 0.0048 \\
\hline
\multicolumn{6}{l}{$n=1000, \qquad \sigma=0.02$}\\
\IgnoreF &  -0.0799 & 0.3099 &  0.0151 & 0.0075 &  0.0147 & 0.0072 \\
Update   &  -0.0795 & 0.3097 &  0.0052 & 0.0051 &  0.0052 & 0.0049 \\
OR       &  -0.0792 & 0.3092 &  0.0050 & 0.0044 &         &        \\
ML       &  -0.0794 & 0.3093 &  0.0048 & 0.0043 &         &        \\
RBAN     &  -0.0797 & 0.3098 &  0.0063 & 0.0057 &  0.0052 & 0.0049 \\
\hline
\end{tabular*}\vspace*{-6pt}
\end{table}

\begin{table}[t]
\def\arraystretch{1.12}
\caption{Means, standard deviations, and median standard errors
of the estimates of intersection points
for discrete distribution of the true points}
\label{table1-Discrete}\vspace*{-9pt}
\begin{tabular*}{\textwidth}
{@{\extracolsep{\fill}}l@{\qquad\quad}D{.}{.}{4}D{.}{.}{4}%
     @{\qquad\quad}D{.}{.}{4}D{.}{.}{4}%
     @{\qquad\quad}D{.}{.}{4}D{.}{.}{4}@{}}
\hline
\multicolumn{1}{@{}c@{\qquad\quad}}{Method}&
\multicolumn{2}{@{}c@{\qquad\quad}}{Means}&
\multicolumn{2}{@{}c@{\qquad\quad}}{\hbox to 0pt{\hss Standard~deviations\hss}}&
\multicolumn{2}{@{}c@{}}{Standard errors}\\
\hline
True value& -0.08 & 0.31 & & & & \\\hline
\multicolumn{6}{l}{$n=1000, \qquad \sigma=0.1$}\\
\IgnoreF &  -0.0699 & 0.3077 &  0.0263 & 0.0290 &  0.0241 & 0.0263 \\
Update   &  -0.0722 & 0.3116 &  0.0197 & 0.0186 &  0.0203 & 0.0175 \\
OR       &  -0.0755 & 0.3188 &  0.0148 & 0.0144 &         &        \\
ML       &  -0.0958 & 0.3315 &  0.0131 & 0.0120 &         &        \\
RBAN     &  -0.0717 & 0.3105 &  0.0209 & 0.0175 &  0.0205 & 0.0178 \\
\hline
\multicolumn{6}{l}{$n=10000, \qquad \sigma=0.1$}\\
\IgnoreF &  -0.0783 & 0.3109 &  0.0092 & 0.0078 &  0.0080 & 0.0083 \\
Update   &  -0.0785 & 0.3114 &  0.0071 & 0.0054 &  0.0065 & 0.0061 \\
OR       &  -0.0721 & 0.3157 &  0.0048 & 0.0046 &         &        \\
ML       &  -0.0931 & 0.3278 &  0.0043 & 0.0035 &         &        \\
RBAN     &  -0.0786 & 0.3113 &  0.0071 & 0.0053 &  0.0065 & 0.0061 \\
\hline
\multicolumn{6}{l}{$n=100000, \qquad \sigma=0.1$}\\
\IgnoreF &  -0.0798 & 0.3098 &  0.0031 & 0.0024 &  0.0026 & 0.0027 \\
Update   &  -0.0799 & 0.3099 &  0.0025 & 0.0016 &  0.0021 & 0.0019 \\
OR       &  -0.0715 & 0.3151 &  0.0017 & 0.0013 &         &        \\
ML       &  -0.0932 & 0.3283 &  0.0013 & 0.0011 &         &        \\
RBAN     &  -0.0799 & 0.3099 &  0.0024 & 0.0017 &  0.0021 & 0.0019 \\
\hline
\multicolumn{6}{l}{$n=1000, \qquad \sigma=0.02$}\\
\IgnoreF &  -0.0796 & 0.3094 &  0.0033 & 0.0032 &  0.0036 & 0.0033 \\
Update   &  -0.0798 & 0.3097 &  0.0030 & 0.0024 &  0.0033 & 0.0023 \\
OR       &  -0.0782 & 0.3086 &  0.0021 & 0.0018 &         &        \\
ML       &  -0.0786 & 0.3087 &  0.0019 & 0.0018 &         &        \\
RBAN     &  -0.0796 & 0.3092 &  0.0030 & 0.0028 &  0.0033 & 0.0024 \\
\hline
\end{tabular*}
\end{table}

\begin{table}[t]
\def\arraystretch{1.12}
\caption{Means, standard deviations, and median standard errors
of the estimates of intersection points
for uniform distribution of the true points on two line segments}
\label{table1-Uniform}\vspace*{-9pt}
\begin{tabular*}{\textwidth}
{@{\extracolsep{\fill}}l@{\qquad\quad}D{.}{.}{4}D{.}{.}{4}%
     @{\qquad\quad}D{.}{.}{4}D{.}{.}{4}%
     @{\qquad\quad}D{.}{.}{4}D{.}{.}{4}@{}}
\hline
\multicolumn{1}{@{}c@{\qquad\quad}}{Method}&
\multicolumn{2}{@{}c@{\qquad\quad}}{Means}&
\multicolumn{2}{@{}c@{\qquad\quad}}{\hbox to 0pt{\hss Standard~deviations\hss}}&
\multicolumn{2}{@{}c@{}}{Standard errors}\\
\hline
True value& -0.08 & 0.31 & & & & \\\hline
\multicolumn{6}{l}{$n=1000, \qquad \sigma=0.1$}\\
\IgnoreF &  -0.0785 & 0.3122 &  0.0363 & 0.0274 &  0.0318 & 0.0301 \\
Update   &  -0.0794 & 0.3140 &  0.0216 & 0.0258 &  0.0205 & 0.0290 \\
OR       &  -0.0616 & 0.3127 &  0.0185 & 0.0167 &         &        \\
ML       &  -0.0934 & 0.3118 &  0.0116 & 0.0111 &         &        \\
RBAN     &  -0.0807 & 0.3126 &  0.0219 & 0.0293 &  0.0193 & 0.0292 \\
\hline
\multicolumn{6}{l}{$n=10000, \qquad \sigma=0.1$}\\
\IgnoreF &  -0.0796 & 0.3107 &  0.0103 & 0.0103 &  0.0099 & 0.0095 \\
Update   &  -0.0796 & 0.3110 &  0.0067 & 0.0103 &  0.0065 & 0.0094 \\
OR       &  -0.0639 & 0.3087 &  0.0064 & 0.0049 &         &        \\
ML       &  -0.0904 & 0.3106 &  0.0042 & 0.0033 &         &        \\
RBAN     &  -0.0797 & 0.3107 &  0.0066 & 0.0104 &  0.0064 & 0.0095 \\
\hline
\multicolumn{6}{l}{$n=100000, \qquad \sigma=0.1$}\\
\IgnoreF &  -0.0798 & 0.3098 &  0.0035 & 0.0030 &  0.0032 & 0.0030 \\
Update   &  -0.0798 & 0.3098 &  0.0021 & 0.0029 &  0.0020 & 0.0030 \\
OR       &  -0.0625 & 0.3085 &  0.0015 & 0.0014 &         &        \\
ML       &  -0.0891 & 0.3107 &  0.0012 & 0.0011 &         &        \\
RBAN     &  -0.0796 & 0.3097 &  0.0023 & 0.0030 &  0.0020 & 0.0030 \\
\hline
\multicolumn{6}{l}{$n=1000, \qquad \sigma=0.02$}\\
\IgnoreF &  -0.0799 & 0.3100 &  0.0041 & 0.0032 &  0.0041 & 0.0035 \\
Update   &  -0.0798 & 0.3101 &  0.0033 & 0.0032 &  0.0032 & 0.0034 \\
OR       &  -0.0803 & 0.3103 &  0.0023 & 0.0021 &         &        \\
ML       &  -0.0805 & 0.3101 &  0.0022 & 0.0020 &         &        \\
RBAN     &  -0.0798 & 0.3100 &  0.0035 & 0.0033 &  0.0032 & 0.0034 \\
\hline
\end{tabular*}
\end{table}
%%%%%%%{tables1.tex}

Using the estimator $\widehat{F}$
(by \emph{one-step update\/} before ignore-$\widehat{F}$ step), we
improve the precision of estimation.
The precision of the \emph{RBAN-moment} estimator approximates
the precision of the \emph{updated before ignore-$\widehat{F}$ step\/} estimator.

The parametric \emph{maximum likelihood} estimator is the best
when the normality condition, which was assumed during
construction of the estimator, is satisfied. Otherwise, it is
biased.

The \emph{orthogonal regression\/} and the \emph{maximum likelihood\/} estimators
are good for small error variance ($\sigma^2 = 0.02^2$).
For $\sigma^2 = 0.1^2$, the \emph{orthogonal regression\/} estimator is broken down
when the distribution of true points is a~mixture of two normal distributions
and is biased for the two other distributions of true
points.\looseness=1

Mean-square deviance of the intersection
of the estimated lines from the true
intersection point is presented in Table~\ref{table-fullDeviances}.
%%%%%%%{table2.tex}
\begin{table}[t!]
\def\arraystretch{1.12}
\caption{Mean-square distances between estimated and true intersection points}
\label{table-fullDeviances}\vspace*{-9pt}
\begin{tabular*}{\textwidth}
{@{\extracolsep{\fill}}rD{.}{.}{2}%
 D{.}{.}{4}%
 D{.}{.}{4}%
 D{.}{.}{4}%
 D{.}{.}{4}%
 D{.}{.}{4}@{}}
\hline
\multicolumn{1}{c}{$n$}&
\multicolumn{1}{c}{$\sigma$}&
\multicolumn{1}{c}{\IgnoreF}&
\multicolumn{1}{c}{Update}&
\multicolumn{1}{c}{OR}&
\multicolumn{1}{c}{ML}&
\multicolumn{1}{c@{}}{RBAN}\\
\hline
\multicolumn{7}{l}{\emph{Distribution of true points is a mixture of normals}}\\
  1000 & 0.1 & 0.9403 & 0.0954 & 3.2978 & 0.0421 & 0.6124\\ % 0.5743
 10000 & 0.1 & 0.0722 & 0.0201 & 2.9127 & 0.0138 & 0.0199\\ % 0.0256
100000 & 0.1 & 0.0230 & 0.0056 & 2.9645 & 0.0038 & 0.0056\\ % 0.0058
1000 & 0.02  & 0.0168 & 0.0073 & 0.0067 & 0.0065 & 0.0084\\ % 0.0073
\hline
\multicolumn{7}{l}{\emph{Discrete distribution of true points}}\\
  1000 & 0.1 & 0.0403 & 0.0281 & 0.0228 & 0.0320 & 0.0284\\ % 0.0287
 10000 & 0.1 & 0.0121 & 0.0091 & 0.0118 & 0.0228 & 0.0090\\ % 0.0091
100000 & 0.1 & 0.0039 & 0.0029 & 0.0101 & 0.0226 & 0.0029\\ % 0.0030
1000 & 0.02  & 0.0046 & 0.0038 & 0.0036 & 0.0032 & 0.0042\\ % 0.0038
\hline
\multicolumn{7}{l}{\emph{Uniform distribution of true points}}\\
  1000 & 0.1 & 0.0453 & 0.0367 & 0.0310 & 0.0209 & 0.0365\\ % 0.0345
 10000 & 0.1 & 0.0145 & 0.0123 & 0.0181 & 0.0117 & 0.0123\\ % 0.0123
100000 & 0.1 & 0.0046 & 0.0036 & 0.0177 & 0.0093 & 0.0037\\ % 0.0036
1000 & 0.02  & 0.0052 & 0.0046 & 0.0031 & 0.0030 & 0.0048\\ % 0.0048
\hline
\end{tabular*}
\end{table}
%%%%%%%{table2.tex}

For small errors, the \emph{RBAN-moment} estimator is a bit less
accurate than the \emph{updated before ignore-$\widehat{F}$
step\/} estimator. For $\sigma^2 = 0.1^2$, the difference is
negligible.

The parametric \emph{maximum likelihood} estimator has the smallest
deviation from the true value, except for the discrete
distribution of true points and $\sigma^2 = 0.01$.

For small errors ($\sigma^2 = 0.02^2$), the \emph{orthogonal
regression\/} estimator outperforms the consistent estimators and
has the deviation approximately as small as the parametric
\emph{maximum likelihood\/} estimator.

Normalization of the estimator of $\boldbeta$
affects the ALS2-based estimator of two lines
with \emph{one-step update before the ignore-$\widehat{F}$ step}.
With normalization $\|\hat{\mathbi{\beta}}\|=1$,
the derived estimator of two lines is not equivariant,
whereas with normalization
$\tilde{\mathbi{\beta}}^\top {\boldsymbol\varPsi}'_n(\hat\sigma^2)
 \tilde{\mathbi{\beta}} = -n$,
the derived estimator is equivariant.
Comparison of equivariant and nonequivariant
versions of the estimator is displayed in
Table~\ref{table-UpdateVers}.
%%%%%%%%{tablesU.tex}
\begin{table}[p]
\def\arraystretch{1.03}
\caption{Comparison of two versions
(equivariant (ev) and nonequivariant (ne))
of the updated before ignore-$\widehat{F}$ step estimator}
\label{table-UpdateVers}\vspace*{-9pt}
\begin{tabular*}{\textwidth}
{@{\extracolsep{\fill}}rD{.}{.}{2}c%
 D{.}{.}{6}D{.}{.}{6}%
 D{.}{.}{6}D{.}{.}{6}%
 D{.}{.}{6}D{.}{.}{6}@{}}
\hline
\multicolumn{1}{@{}c}{$n$}&
\multicolumn{1}{c}{$\sigma$}&
Ver.&
\multicolumn{2}{c}{Means}&
\multicolumn{2}{c}{Standard deviations}&
\multicolumn{2}{c@{}}{Standard errors}\\
\hline
\multicolumn{3}{l}{True value:}& -0.08 & 0.31 & & & & \\\hline
\multicolumn{9}{l}{\emph{Distribution of true points is a mixture of normals}}\\
  1000 & 0.1 & ev & -0.082046 & 0.291175 & 0.070617 & 0.062003 & 0.043713 & 0.047939\\
       &     & ne & -0.044038 & 0.247382 & 0.251372 & 0.514473 & 0.038853 & 0.050167\\
 10000 & 0.1 & ev & -0.079623 & 0.308039 & 0.012652 & 0.015582 & 0.012403 & 0.015471\\
       &     & ne & -0.085055 & 0.304177 & 0.015924 & 0.018695 & 0.012506 & 0.015550\\
100000 & 0.1 & ev & -0.080137 & 0.309780 & 0.003710 & 0.004173 & 0.003925 & 0.004749\\
       &     & ne & -0.080991 & 0.309386 & 0.003880 & 0.004255 & 0.003926 & 0.004742\\
1000 & 0.02  & ev & -0.079548 & 0.309703 & 0.005156 & 0.005131 & 0.005174 & 0.004891\\
       &     & ne & -0.079918 & 0.309508 & 0.005247 & 0.005149 & 0.005179 & 0.004918\\
\hline
\multicolumn{9}{l}{\emph{Discrete distribution of true points}}\\
  1000 & 0.1 & ev & -0.072202 & 0.311648 & 0.019709 & 0.018553 & 0.020266 & 0.017500\\
       &     & ne & -0.071460 & 0.312230 & 0.020049 & 0.018740 & 0.020213 & 0.017457\\
 10000 & 0.1 & ev & -0.078482 & 0.311377 & 0.007087 & 0.005371 & 0.006518 & 0.006066\\
       &     & ne & -0.078418 & 0.311436 & 0.007090 & 0.005387 & 0.006520 & 0.006054\\
100000 & 0.1 & ev & -0.079868 & 0.309929 & 0.002460 & 0.001647 & 0.002060 & 0.001900\\
       &     & ne & -0.079863 & 0.309934 & 0.002461 & 0.001647 & 0.002060 & 0.001901\\
1000 & 0.02  & ev & -0.079772 & 0.309728 & 0.002967 & 0.002376 & 0.003320 & 0.002344\\
     &       & ne & -0.079755 & 0.309732 & 0.002963 & 0.002375 & 0.003319 & 0.002339\\
\hline
\multicolumn{9}{l}{\emph{Uniform distribution of true points}}\\
  1000 & 0.1 & ev & -0.079405 & 0.313977 & 0.021551 & 0.025759 & 0.020451 & 0.028992\\
       &     & ne & -0.078507 & 0.315350 & 0.022219 & 0.026091 & 0.020512 & 0.029115\\
 10000 & 0.1 & ev & -0.079604 & 0.311024 & 0.006673 & 0.010337 & 0.006467 & 0.009389\\
       &     & ne & -0.079576 & 0.311176 & 0.006685 & 0.010349 & 0.006456 & 0.009372\\
100000 & 0.1 & ev & -0.079795 & 0.309802 & 0.002075 & 0.002919 & 0.001974 & 0.002994\\
       &     & ne & -0.079794 & 0.309818 & 0.002076 & 0.002921 & 0.001972 & 0.002992\\
1000 & 0.02  & ev & -0.079833 & 0.310081 & 0.003252 & 0.003249 & 0.003172 & 0.003418\\
       &     & ne & -0.079825 & 0.310097 & 0.003250 & 0.003249 & 0.003169 & 0.003411\\
\hline
\end{tabular*}
\end{table}
%%%%%%%%{tablesU.tex}

%%%%%%%{Table2ACM.tex}
\begin{table}[p]
\def\arraystretch{1.03}
\tabcolsep=0pt
\caption{Coverage probability and area
of confidence ellipsoids (c.e.) for centers
by the ALS2 estimator}\label{tab:table2ACM}\vspace*{-9pt}
\begin{tabular*}{\textwidth}
{@{\extracolsep{\fill}}rlcccccc@{}}
\hline
\multicolumn{1}{c}{$n$} &
\multicolumn{1}{c}{$\sigma$} &
\multicolumn{3}{c}{$\widehat{\varSigma}_{\rm true}$-based estimator} &
\multicolumn{3}{c}{$\widehat{\varSigma}_{\rm sample}$-based estimator} \\
\cline{3-5}\cline{6-8}
&& \multicolumn{2}{c}{Coverage probab.}
& \multicolumn{1}{l}{Area of}
& \multicolumn{2}{c}{Coverage probab.}
& \multicolumn{1}{l}{Area of} \\
\cline{3-4}\cline{6-7}
&&
\multicolumn{1}{c}{80\%,} &
\multicolumn{1}{c}{95\%,} &
\multicolumn{1}{l}{95\% c.e.,} &
\multicolumn{1}{c}{80\%,} &
\multicolumn{1}{c}{95\%,} &
\multicolumn{1}{l}{95\% c.e.,} \\
&&
\multicolumn{1}{c}{\%} &
\multicolumn{1}{c}{\%} &
\multicolumn{1}{l}{$\times 10^{-4}$} &
\multicolumn{1}{c}{\%} &
\multicolumn{1}{c}{\%} &
\multicolumn{1}{l}{$\times 10^{-4}$} \\
\hline
\multicolumn{8}{l}
{\emph{Distribution of true points is a mixture of normals}} \\
  1000 & 0.1 & 70.6 & 80.2 & 1449. & 70.0 & 79.2 & 1562. \\
 10000 & 0.1 & 79.4 & 93.8 & 236.2 & 79.6 & 92.9 & 259.4 \\
100000 & 0.1 & 80.7 & 94.9 & 15.38 & 80.6 & 94.9 & 15.17 \\
1000 &  0.02 & 80.4 & 95.1 & 15.95 & 78.0 & 94.1 & 19.86 \\
\hline
\multicolumn{8}{l}{\emph{Discrete distribution of true points}} \\
  1000 & 0.1 & 78.1 & 93.9 & 81.80  & 77.4 & 93.4 & 78.94  \\
 10000 & 0.1 & 81.1 & 95.6 & 12.34  & 80.9 & 95.8 & 12.39  \\
100000 & 0.1 & 80.1 & 94.6 &  1.205 & 79.9 & 94.7 &  1.204 \\
1000 &  0.02 & 81.0 & 94.9 &  1.984 & 81.3 & 95.2 &  1.988 \\
\hline
\multicolumn{8}{l}{\emph{Uniform distribution of true points}} \\
  1000 & 0.1 & 82.1 & 94.3 & 152.9 & 81.5 & 94.3 & 138.4 \\
 10000 & 0.1 & 81.0 & 96.6 & 20.36 & 80.3 & 96.3 & 18.92 \\
100000 & 0.1 & 78.4 & 95.0 & 1.823 & 78.6 & 95.0 & 1.842 \\
1000 &  0.02 & 78.7 & 94.7 & 2.926 & 78.5 & 94.1 & 3.041 \\
\hline
\end{tabular*}
\end{table}
%%%%%%%{Table2ACM.tex}

There is a tendency that the equivariant version of the estimator
is more accurate for small samples than
the nonequivariant version.
The two versions of the estimator are
consistent and asymptotically equivalent.
When the estimation is precise, the difference between the versions
is negligible.
When the estimation is imprecise, it is impossible to make inference which
version is more accurate.

\subsection{Comparison of two estimators for asymptotic
covariance matrix in the conic section fitting
model}

In \cite{Shklyar2015p2} a conic section fitting model is
considered, and two estimators
($\widehat\varSigma_{\rm true}$ and
$\widehat\varSigma_{\rm sample}$)  for the
asymptotic covariance matrix of the ALS2 estimator
are constructed.

The software developed here can be used to make numerical
comparison of the estimates of the asymptotic covariance
matrices.
The data are generated as described in
Section~\ref{ss:SimulationSetup}
with 1000 simulations for each set of true points.
Thus, the true conic unnecessarily was chosen
degenerate.
For each simulation, the parameters of the conic
section were estimated; its center is found,
and two confidence ellipsoids for the center were
constructed using two different estimators of the
asymptotic covariance matrix.

The sample coverage probability and median (over
1000 ellipsoids) area of the confidence ellipsoids
is presented in Table~\ref{tab:table2ACM}.  The ellipsoids were
constructed for confidence levels 0.8 and 0.95.
The area of 95\% confidence ellipsoids is displayed
in Table~\ref{tab:table2ACM},
and the area of 80\% confidence ellipsoid is
$\log_{20}(5) = 0.5372$
of the area of 95\% confidence ellipsoids.

Note that standard errors for coverage probability
are $1.3\%$ for 80\% confidence ellipsoids and
$0.7\%$ for 95\% confidence ellipsoids.
The simulations do not allow us to make an inference which
estimator is better.
Thus, $\widehat{\varSigma}_{\rm sample}$-based estimator
\emph{updated before ignore-$\widehat{F}$ step\/}
is compared with other estimators in simulations
in Section~\ref{ss:simresults}
because of simpler explicit expression
for $\widehat{\varSigma}_{\rm sample}$.

%% Appendices %%
%%%%%%%%%%%%%%%%%%%%%%
\appendix
\section{Proofs}
\label{Appendix}

\begin{proof}[Proof of Proposition~\ref{prop:corec1}.]
The strong consistency of the estimator follows from \cite[Theorem 17]{Shklyar2007}.
Under the conditions of Proposition~\ref{prop:corec1},
\begin{gather}
\label{eq:convEP1} \frac{1}n {\boldsymbol\varPsi}_n\bigl(
\sigma^2\bigr) \to \Psiinf \quad \mbox{a.s.}
\\
\label{eq:convEP2} \frac{1}n \overline{\boldsymbol\varPsi}_n \to
\Psiinf \quad \mbox{(a.s. in the structural model).}
\end{gather}
By Lemma 5 in \cite{Shklyar2007},
$\limsup \frac{1}n \boldbeta^\top \overline{\boldsymbol\varPsi}'_n  \boldbeta < 0$,
which, together with \eqref{eq:convEP2}, implies $\boldbeta^\top \dPsiinf \boldbeta < 0$
(see the proof of Theorem 2 in \cite{Shklyar2015p1}).
Then
\begin{equation}
\label{eq:convergdenomQ} \frac{\hatboldbeta^\top {\boldsymbol\varPsi}'_n(\hat\sigma^2) \hatboldbeta}{
n\,\|\hatboldbeta\|^2} = \frac{\hatboldbeta^\top {\boldsymbol\varPsi}'_n(\sigma^2) \hatboldbeta}{
n\,\|\hatboldbeta\|^2} + \frac{(\hat\sigma^2 - \sigma^2)
\hatboldbeta^\top \ddPsi \hatboldbeta}{
\|\hatboldbeta\|^2} \to
\frac{\boldbeta^\top \dPsiinf \boldbeta}{
\|\boldbeta\|^2} \quad\mbox{a.s.}
\end{equation}
Eventually, the left-hand side of \eqref{eq:convergdenomQ} in negative.
\end{proof}

\begin{proof}[Proof of Proposition~\ref{prop:propALS2}]
The strong consistency of $\widetilde{\boldbeta}$ follows from
%convergences
\eqref{eq:consHatbeta1} and \eqref{eq:convergdenomQ}.
The proof of asymptotic normality and consistency of the estimator
of the asymptotic covariance matrix can be obtained
by modification of the proofs of Theorem 2 in \cite{Shklyar2015p1}
and Theorem 3 in \cite{Shklyar2015p2}.
\end{proof}

\begin{proof}[Proof of Proposition~\ref{prop:2lcALS2f}]
The conditions of consistency Theorem~1 in \cite{Shklyar2015p1}
can be verified, and the consistency follows.

The most tedious is the condition
\begin{equation}
\liminf_{n\to\infty} \frac{1}n \lambda_{\min,2} (
\overline \varPsi_n) > 0 . \label{eq:scondiTh1p1}
\end{equation}

Denote
\[
K_j = \begin{pmatrix}
0 & 0 & 1 \\
0 & h_j & k_j \\
0 & 1 & 0 \\
h_j & k_j & 0 \\
1 & 0 & 0
\end{pmatrix}, \quad j=1, \, 2.
\]
Then $K_1^{} K_1^\top + K_2^{} K_2^\top$ is a positive semidefinite
matrix, and
\[
\det \bigl(K_1^{} K_1^\top +
K_2^{} K_2^\top\bigr) = 2
(h_1 - h_2)^4 + 2(h_1 -
h_2)^2 (k_1 - k_2)^2
+ 2 (k_1 - k_2)^4 > 0 .
\]\eject
\noindent Thus, $\lambda_{\min} (K_1^{} K_1^\top + K_2^{} K_2^\top) > 0$.

The matrix
\[
K_1^{} \sum_{\substack{i=1,\ldots,n\\ \nu(i)=1}} \begin{pmatrix}
1        & \xi_i^{} & \xi_i^2 \\
\xi_i^{} & \xi_i^2  & \xi_i^3 \\
\xi_i^2  & \xi_i^3  & \xi_i^4
\end{pmatrix}
K_1^\top + K_2 \sum
_{\substack{i=1,\ldots,n\\ \nu(i)=2}} \begin{pmatrix}
1        & \xi_i^{} & \xi_i^2 \\
\xi_i^{} & \xi_i^2  & \xi_i^3 \\
\xi_i^2  & \xi_i^3  & \xi_i^4
\end{pmatrix} K_2^\top
\]
is the principal submatrix of $\overline{\varPsi}_n$.
By the Cauchy interlacing theorem,
\begin{align*}
\lambda_{\min,2} (\overline{\varPsi}_n) &\ge
\lambda_{\min} \left( \sum_{j=1}^2
K_j^{} \sum_{\substack{i=1,\ldots,n\\ \nu(i)=j}} \begin{pmatrix}
1        & \xi_i^{} & \xi_i^2 \\
\xi_i^{} & \xi_i^2  & \xi_i^3 \\
\xi_i^2  & \xi_i^3  & \xi_i^4
\end{pmatrix}
K_j^\top \right)
\\
&\ge \lambda_{\min} \bigl(K_1^{}
K_1^\top + K_2^{}
K_2^\top\bigr) \min_{j=1,\, 2} \left(
\lambda_{\min} \left( \sum_{\substack{i=1,\ldots,n\\ \nu(i)=j}} \begin{pmatrix}
1        & \xi_i^{} & \xi_i^2 \\
\xi_i^{} & \xi_i^2  & \xi_i^3 \\
\xi_i^2  & \xi_i^3  & \xi_i^4
\end{pmatrix}
\right) \right) ,
\end{align*}
and then inequality \eqref{eq:scondiTh1p1} can easily be proved.
\end{proof}

\begin{proof}[Proof of Proposition~\ref{prop:2lcALS2s}.]
Proposition~\ref{prop:2lcALS2s} follows from Proposition 25 in
\cite{Shklyar2007}.
The identifiability condition (S5-) in \cite{Shklyar2007} holds
because the intersection of a~couple of lines and a conic section
may be a finite set with not more than four points,
a straight line,
a straight line and a point outside the line,
or the couple of lines and the conic section coincide;
in the last case, the coefficients of the equations
for the lines and the conic section satisfy relations
\eqref{eqs:2linesTOconic}.
\end{proof}

\begin{proof}[Proofs of Propositions~\ref{prop:ignF_func_cons},
\ref{prop:ignF_struct_cons},
\ref{prop:ignF_func_AN}, and
\ref{prop:ignF_struct_AN}.]
Consistence of the ``ignore-$\widehat{F}$'' estimator
follows from the consistency of
the ALS2 estimator $\hat{\mathbi{\beta}}$ and
from the continuity of the function $\betatolines(\mathbi{\beta})$ at the point
of the true value of the parameter $\mathbi{\beta}$.
The asymptotic normality of the ``ignore-$\widehat{F}$''
follows from the asymptotic normality of
$\tilde{\mathbi{\beta}}$ and
the differentiability of
$\betatolines(\mathbi{\beta})$ at the point
$\betanormaltrue$.
\end{proof}

\begin{proof}[Proof of Proposition~\ref{prop:can_1st}.]
The consistency and asymptotic normality of
the $\tilde{\mathbi{\beta}}$ estimator,
the differentiability of the functional $\Delta(\mathbi{\beta})$
at point $\betanormaltrue$ (note that $\Delta(\betanormaltrue) = 0$),
and the convergence
\[
\frac{1}{\Delta'(\tilde{\mathbi{\beta}}) \widehat{\varSigma}_{\tilde\beta}
\Delta'(\tilde{\mathbi{\beta}})^\top} \widehat{\varSigma}_{\tilde\beta} \Delta'(
\tilde{\mathbi{\beta}})^\top \to \frac{1}{\Delta'(\betanormaltrue) \varSigma_{\tilde\beta}
\Delta'(\betanormaltrue)^\top} \varSigma_{\tilde\beta}
\Delta'(\betanormaltrue)^\top
\]
imply the convergence and asymptotic normality of
the updated estimator~$\tilde{\mathbi{\beta}}_{\rm 1st}$.
Thus, the consistency and asymptotic normality
of $(\hat k_{1,{\rm 1st}},
 \hat h_{1,{\rm 1st}},
 \hat k_{2,{\rm 1st}},
 \hat h_{2,{\rm 1st}})^\top$ can be\break proved
similarly to those of the ``ignore-$\widehat{F}$'' estimator.
\end{proof}

%% Bibliography %%
%%%%%%%%%%%%%%%%%%%%%%
%\bibliography{bib/2lines}

%\input{tablesU.tex}
\end{document}